\documentclass[journal,onecolumn,twoside,web]{ieeecolor}
\usepackage{generic}

\usepackage{textcomp}

\usepackage{graphicx}
\usepackage{epstopdf}
\usepackage{multicol}
\usepackage{stfloats}
\usepackage{amsfonts}

\usepackage{mathrsfs}
\usepackage{multicol}
\usepackage{stfloats}
\usepackage{amsfonts}
\usepackage{mathrsfs}
\usepackage{amsfonts}
\usepackage{empheq}
\usepackage[linesnumbered,ruled,lined,boxed]{algorithm2e}
\usepackage[procnames]{listings}
\usepackage{algpseudocode}
\usepackage{amsthm}
\usepackage{color}
\usepackage[ruled]{algorithm2e}
\usepackage{fancybox}
\usepackage{framed}
\usepackage{amsfonts, color}
\usepackage{setspace}
\usepackage{nicematrix}
\usepackage[dvips]{epsfig}
\usepackage{framed}
\usepackage{psfrag, amssymb, amsmath, cite}
\usepackage[colorlinks,linkcolor=blue,anchorcolor=blue,citecolor=blue]{hyperref}
\DontPrintSemicolon

\usepackage{tikz}
\usetikzlibrary{fit}
\tikzset{%
	highlight/.style={rectangle,rounded corners,fill=red!15,draw,
		fill opacity=0.5,thick,inner sep=0pt}
}

\usepackage{verbatim}

\newtheorem{remark}{Remark}
\newtheorem{definition}{\textbf{Definition}}

\newtheorem{theorem}{\textbf{Theorem}}

\newtheorem{assumption}{\textbf{Assumption}}
\newtheorem{lemma}{\textbf{Lemma}}

\newtheorem*{proof*}{\textbf{Proof}}

\def\BibTeX{{\rm B\kern-.05em{\sc i\kern-.025em b}\kern-.08em
		T\kern-.1667em\lower.7ex\hbox{E}\kern-.125emX}}
\begin{document}
	\title{Cooperative constrained motion coordination \\ of networked heterogeneous vehicles}
	\author{Zhiyong Sun, \IEEEmembership{Member, IEEE}, Marcus Greiff, \IEEEmembership{Student Member, IEEE}, Anders Robertsson, \IEEEmembership{Senior Member, IEEE}, Rolf Johansson, \IEEEmembership{Fellow, IEEE}, and Brian D. O. Anderson, \IEEEmembership{Life Fellow, IEEE},
		\thanks{The research leading to these results has received funding from
			the Swedish Science Foundation (SSF) project ``Semantic mapping
			and visual navigation for smart robots'' (RIT15-0038), and a starting grant from Eindhoven Artificial Intelligence Systems Institute (EAISI), The Netherlands. }
		\thanks{Z. Sun is with Department of Electrical Engineering, Eindhoven University of Technology (TU/e), and also with Eindhoven Artificial Intelligence Systems Institute, Eindhoven, The Netherlands. Email: {\tt\small z.sun@tue.nl, sun.zhiyong.cn@gmail.com. } }
		\thanks{M. Greiff, A. Robertsson and R. Johansson are members of the LCCC Linnaeus Center and the ELLIIT Excellence Center at Lund University, Sweden. Emails: 
			{\tt\small \{marcus.greiff, rolf.johansson, anders.robertsson\}@control.lth.se}}
		\thanks{B. D. O. Anderson is with  School of  Engineering, The Australian National University, Acton, ACT 2601,   Australia. Email: 
			{\tt\small brian.anderson@anu.edu.au}}        
	}
	
	\maketitle
	
	\begin{abstract}
		We consider the problem of cooperative motion coordination   for multiple  heterogeneous mobile vehicles subject to various constraints. These include nonholonomic motion constraints, constant speed constraints, holonomic coordination constraints, and equality/inequality geometric constraints. We develop a general framework involving differential-algebraic equations and viability theory to  determine coordination feasibility for a coordinated motion control under heterogeneous vehicle dynamics and different types of coordination task constraints. 
		If a coordinated motion solution exists for the derived differential-algebraic equations and/or inequalities, a constructive algorithm is proposed to derive an equivalent dynamical system that  generates a set of feasible coordinated motions for each individual vehicle. In case studies on coordinating two vehicles, we derive analytical solutions to motion generation for two-vehicle groups consisting of car-like vehicles, unicycle vehicles, or vehicles with constant speeds, which serve as benchmark coordination tasks for more complex vehicle groups. The motion generation algorithm is well-backed by simulation data for a wide variety of coordination situations involving heterogeneous vehicles. 
		We then extend the  vehicle  control framework to deal with the cooperative coordination problem with time-varying coordination tasks and leader-follower structure. 
		We show several simulation experiments on multi-vehicle coordination under various constraints to validate the theory and the effectiveness of the proposed schemes. 
	\end{abstract}
	
	\begin{IEEEkeywords}
		Multi-vehicle coordination; networked heterogeneous systems;   time-varying constraints; affine distribution and codistribution; controlled-invariant set; unicycle vehicle; constrained motion control. \\
		
		Accompanying simulation videos are available at \url{https://youtu.be/wFx6x7Mep74}. 
	\end{IEEEkeywords}
	
	
	\section{Introduction}
	\subsection{Background and motivation}
	\IEEEPARstart{I}{n} the active research field of mobile robot/vehicle motion planning and control, multi-vehicle motion coordination and cooperative control have been and will remain attractive research topics, motivated by an increasing number of practical applications requiring multiple robots or vehicles to cooperatively perform coordinated tasks \cite{knorn2016overview,dames2017detecting,schwager2017multi}. These include multi-robot formation control, area coverage and surveillance, coordinated target tracking, to name a few \cite{oh2015survey,liu2018distributed}. A fundamental problem in multi-vehicle coordination is to plan feasible motion schemes and trajectories for each individual vehicle which should satisfy both kinematic or dynamic requirements for all vehicles, and inter-vehicle geometric constraints that describe the nature of a given coordination task. Typically, an individual vehicle is subject to various   motion constraints, e.g., determined by fundamental physics such as Newton's laws, together with constraints inherent to
	the vehicle on e.g., speed, rate of linear or angular acceleration,  etc. These  constraints limit possible motion directions, while a coordinated motion to achieve a predefined coordination task then further imposes inter-vehicle motion constraints, related to, for example, vehicle spacing, whether by equality constraints or inequality constraints. Recent advances on  heterogeneous multi-robot system \cite{rizk2019cooperative} have revealed  the great potential of cooperative control involving heterogeneous autonomous systems in achieving complex tasks, while the heterogeneity enhances coordination performance for a networked robotic group but also  makes the coordination control a more challenging problem. Till now, how to secure  cooperative constrained motion coordination involving networked heterogeneous systems still remains an open problem, which motivates the research in this paper.

	\subsection{Related work}
	\subsubsection{Cooperative feasible coordination}
	Coordination tasks with multiple mobile vehicles often involve  various types of inter-vehicle constraints, typically described by equality or inequality functions of inter-vehicle geometric variables. For example, a practical coordinated motion may be described by some inequality constraints that require a bounded inter-vehicle distance between mobile vehicles; i.e., a lower bound to guarantee collision avoidance, and an upper bound to avoid communication loss due to excessively long ranges. Furthermore, in multi-robot visibility maintenance control, which requires vehicles' headings to lie in a bounded cone of field of view, coordination constraints are modelled by 
	certain leader-follower geometric inequality functions. All these practical coordination control scenarios call for a general framework for multi-vehicle coordination planning and control under various constraints.
	
	The first key question on cooperative multi-vehicle coordination is to determine \textit{coordination task feasibility}. We say a coordination task is \textit{feasible}, if   dynamically feasible and cooperative motions are available for all vehicles that satisfy both vehicles' kinematic motion constraints and various geometric constraints in the coordination task.  
	The seminal paper by Tabuada \textit{et al.} \cite{tabuada2005motion}   studied the motion feasibility problem in the context of  multi-vehicle formation control. Via the tools of differential geometry, feasibility conditions were derived for a group of mobile vehicles to maintain formation specifications (described by strict equality constraints) in each vehicle's kinematic  motions. Recently, the motion feasibility problem in multi-vehicle formation and cooperative control has attracted renewed interest  in the control and robotics community. The paper \cite{maithripala2011geometric} discusses coordination control with dynamically feasible vehicle motions, and solves a rigid formation shape maintenance task and formation reconfiguration problem. 
	Our preliminary work  \cite{Sun2016feasibility} has investigated the formation and coordination feasibility with mobile robotic systems modelled by control affine nonlinear systems with drift terms (which include fully-actuated systems, under-actuated systems, and nonholonomic vehicles). More recently, the work by Colombo and Dimarogonas  \cite{colombo2018motion} extends the motion feasibility condition in \cite{tabuada2005motion} to multi-vehicle formation control systems on Lie groups. Cooperative transport control using multiple autonomous vehicles can also be formulated as a motion feasibility problem. In a recent paper  \cite{hajieghrary2017cooperative}, the authors discussed cooperative transport of a buoyant load using two autonomous surface vehicles (ASVs) via a differential geometric approach. The ASV's dynamics are described by the standard unicycle-type equations with nonholonomic constraints, while any two joint-vehicles assume a cooperative task to maintain a fixed distance between them. 
	
	We remark that the above referenced papers \cite{tabuada2005motion,maithripala2011geometric,Sun2016feasibility,colombo2018motion,hajieghrary2017cooperative} only discussed formation or coordination control for multiple vehicles with \textit{strict equality} functions. This paper will focus on a more general problem in multi-vehicle coordination control that also includes inequality constraints, and a mix of equality and inequality constraints, as well as time-varying constraints that can describe more complex coordination tasks.   
	
	\subsubsection{Heterogeneous multi-robot systems}
	Cooperative control of multiple heterogeneous multi-robot systems
	has received particular interest recently in the control and robotics community, as it is believed that the heterogeneity of robotic systems improves the performance and versatility in achieving a cooperative task, since different autonomous systems can complement  each other for complex tasks that may not be achieved by homogeneous systems. Recent studies on heterogeneous multi-robotic systems have focused on typical coordination tasks such as swarm \cite{  dorigo2013swarmanoid,    prorok2017impact}, task allocation and routing plan \cite{ sakamoto2020routing,   mayya2021resilient, whitzer2020coordinating}, and delivery tasks \cite{mathew2015planning} involving autonomous heterogeneous UAV-UGV coordination.   It has been further reviewed in \cite{ayanian2019dart} that diversity and heterogeneity can greatly enhance autonomy and cooperative  missions in a robot team. However, since heterogeneous robotic systems comprise individual systems of different kinematics and various motion constraints,  coordinated motion is more challenging to treat since all individual motion constraints should be satisfied. This motivates the research in this paper which aims to develop a general framework to address heterogeneous 
	cooperative multi-vehicle tasks.

	\subsubsection{Autonomous vehicles with motion constraints}
	The problem of maintaining holonomic equality constraints in autonomous  vehicle coordination is also relevant to the framework of virtual holonomic control (VHC). VHC involves a relation (usually described by an equality constraint) among the configuration variables of a mechanical or robotic system which does not physically exist \cite{freidovich2008periodic,    shiriaev2007virtual}. Such constraints are controlled to be invariant via feedback controllers \cite{maggiore2013virtual}. In this paper, we present a \textit{multi-vehicle} framework that includes equality constraints as a special case, and develop admissible control inputs that maintain both equality and inequality constraints in coordination tasks while satisfying nonholonomic motion constraints.  In general,  nonholonomic motion constraints arise from individual vehicle's motion equations, and holonomic constraints come from the coordination tasks described by equality functions or geometric inequalities. 
	Our tools to solve the feasible coordination problem for multiple vehicles with various constraints are an interplay of differential geometry for nonlinear control \cite{isidori1995nonlinear}, viability theory \cite{aubin2009viability} and differential-algebraic equations and inequalities. One of the key tools to address feasible coordination and motion generation with \textit{inequality} motion/coordination constraints is  viability theory \cite{aubin2009viability}, which has relevance in set-invariance control \cite{blanchini1999set}  (the term controlled-invariance set is also used in this context).   It has been used in solving coordination control problem for under-actuated autonomous vehicles in \cite{panagou2014cooperative}, autonomous vehicle  racing control  in \cite{liniger2017real} and visibility maintenance for multiple robotic systems in \cite{panagou2013viability}. 
	
	\subsection{Contributions}
	In this paper, a synthesis of coordination control that respects vehicles' kinematic constraints (which often feature nonholonomic and under-actuated motion constraints) and inter-vehicle constraints (which include holonomic formation constraints, inequality functions or a mix of various constraints) will be provided. To determine the coordination feasibility, a condition in the form of a set of general differential-algebraic equality with inequalities  will be derived that incorporates both vehicle kinematics and geometric coordination functions.  We will further devise an efficient algorithm to solve the proposed feasibility  equations and inequalities  that generate feasible trajectories for all vehicles to achieve a coordination task. We will consider two typical modellings for multiple vehicle coordination control, one based on undirected graph and the other based on leader-follower directed graph. In both cases we present feasibility conditions for vehicle coordination;  feasible motions and vehicle trajectories, if they exist,  can be generated by the devised  algorithm. 
	
	To address time-varying coordination tasks, we will develop a general theory on temporal viability and controlled temporal invariant set. The multi-vehicle control framework is then extended to deal with constrained motion coordination with time-varying geometric task functions. To illustrate the proposed coordination framework and motion generation theory, we also present several application examples on coordinating two or more vehicles with  heterogeneous kinematics and with equality or inequality coordination constraints between inter-vehicle distances or headings. For case studies of typical vehicle models, we derive analytical solutions to motion generation for two-vehicle groups consisting of car-like vehicles, unicycle vehicles, or vehicles with constant (but non-identical) speeds, with motion models in ambient spaces of distinct dimensions, which serve as benchmark coordination tasks for more complex vehicle groups. 
	
	Some preliminary results are presented in conference papers \cite{Sun2016feasibility, sun2019feasible}, which mainly dealt with simple coordination tasks for mobile vehicles. Apart from consolidating these previous publications,  the present paper aims to provide a coherent study to solve the heterogeneous vehicle motion coordination problem, and present a comprehensive framework on constrained vehicle coordination with various (time-invariant or time-varying) tasks and improved motion generation algorithms. 
	
	\subsection{Paper organization}
	This paper is organized as follows. Section~\ref{Sec:preliminary} provides preliminary knowledge of differential geometry, affine distribution/codistribution and vehicle models, and introduces the problem formulation for heterogeneous multi-vehicle coordination.  In Section~\ref{sec:constraints_formulation} we formulate motion constraints arising from vehicles' kinematics and coordination tasks in a unified framework. Section~\ref{sec:coordination_feasibility} presents two key theorems to determine coordination feasibility and devises a constructive algorithm for coordinated motion generation for heterogeneous vehicles. Case studies and examples on coordinating two or more homogeneous and heterogeneous vehicles are shown in Section~\ref{sec:case_study}. Cooperative vehicle coordination under time-varying task functions and leader-follower structures is presented in Section~\ref{sec:time_varying_coordination} and Section~\ref{sec:leader_follower_coordination}, respectively. 
	Concluding remarks in Section~\ref{sec:conclusions} close this paper. 
	
	
	\section{Preliminary and problem formulation} \label{Sec:preliminary}
	In this section we briefly introduce a few concepts and tools to solve the constrained coordination problem, which mainly include  differential geometry approaches on nonlinear control systems from \cite{nijmeijer1990nonlinear, murray1994mathematical} and set-invariance control from \cite{aubin2009viability,blanchini2008set}. 
	
	\subsection{Affine distribution, codistribution and vehicle models}
	A distribution  $\Delta(x)$ on $\mathbb{R}^n$ is an assignment of a linear subspace
	of $\mathbb{R}^n$ at each point $x$. Given a   set of $k$ smooth vector fields $X_1(x), X_2(x), \cdots, X_k(x)$, we define the smooth distribution as $$\Delta(x) = \text{span}\{X_1(x), X_2(x), \cdots, X_k(x)\}.$$ A vector field $X$ belongs to a distribution $\Delta$ if $X(x) \in \Delta(x)$, $\forall x \in \mathbb{R}^n$. A point $x \in\mathbb{R}^n$  is said to be a \textit{regular} point of the distribution $\Delta(x)$ if there exists a neighborhood $U$ of $x$ such that, $\forall \bar x \in U$, the associated distribution $\Delta(\bar x)$ is nonsingular (i.e., of full column rank). In this paper we suppose all smooth distributions under consideration have a constant full column rank and the points associated with a distribution are regular.

	A codistribution assigns to each $x$ a subspace in the dual space denoted by $(\mathbb{R}^{n})^\star$. Given a distribution $\Delta$,  we construct the associated codistribution $\Delta^\perp$ by the annihilator of $\Delta$, which is the set of all covectors that annihilates all vectors in   $\Delta(x)$
	(see \cite[Chapter 1]{isidori1995nonlinear})
	$$\Delta^\perp(x) = \{\omega(x) \in  (\mathbb{R}^{n})^\star  | \left \langle \omega(x),   X \right \rangle = 0, \,\, \forall X \in \Delta \}.$$
	Suppose $x$ is a regular point of a smooth distribution $\Delta$. Then $x$ is also a regular point of the associated codistribution $\Delta^\perp$, where $\Delta^\perp$ is smooth with a constant full rank. Furthermore, $\Delta^\perp$ can be spanned by a set of smooth covectors that are linearly independent at each point in the neighborhood of $x$ \cite{isidori1995nonlinear}. 
	
	In this paper, we model each individual vehicle's motion equations  by the following general form (i.e., control-affine system)
	\begin{align} \label{eq:system_drift}
		\dot p_i = f_{i,0}(p_i) + \sum_{j=1}^{l_i} f_{i,j}(p_i) u_{i,j},
	\end{align}
	where $p_i \in    \mathbb{R}^{n_i}$ is the state of vehicle $i$  (where $n_i$ denotes the dimension of state space for vehicle $i$), $f_{i,0}$ is a smooth drift term, and $u_{i,j}$ is the \emph{scalar} control input associated with the smooth vector field $f_{i,j}$, and $l_i$ is the number of vector field functions.
	Such a nonlinear control-affine system~\eqref{eq:system_drift} with a drift term is very general in that it describes many different types of real-life vehicle kinematics and motion control systems, including fully actuated vehicles, under-actuated systems, and vehicle models subject to holonomic/nonholonomic motion constraints (see the several typical vehicle examples provided in Section~\ref{sec:kinematic_model}).
	
	
	\subsection{Viability theory and set-invariance control}
	In this paper, we will provide a general framework to solve complex coordination problems with mixed equality/inequality constraints. A key tool to address inequality constraints is   viability theory and set-invariance control \cite{blanchini1999set,aubin2009viability}. We now introduce some background knowledge, concepts and theorems on viability theory and set-invariance control. 
	
	\begin{definition}
		(\textbf{Viability and viable set}, \cite{aubin2009viability}) Consider a control system described by a differential equation $\dot x(t) = f(x(t))$ with the state $x \in \mathbb{R}^n$. A subset $\mathcal{F}$ of $\mathbb{R}^n$ enjoys the viability property for the system $\dot x(t) = f(x(t))$  if for every initial state $x(0) \in \mathcal{F}$, there exists   at least one solution to the system
		starting at $x(0)$ which is viable in the time interval $[0, \bar t\,)$ (where $\bar t >0$ denotes the longest time instant of extending solutions) in the sense that $$\forall t \in [0, \bar t \,), x(t) \in \mathcal{F}.$$
	\end{definition}
	
	
	Now define a distance function for a point $y$ to a closed set $\mathcal{F}$ as $d_\mathcal{F}(y)=:  \inf\limits_{z \in \mathcal{F}} \|y-z\|$ (where $\|\cdot\|$ denotes a relevant norm). We recall the notion of a  contingent cone \cite{aubin2009viability} as follows. 
	\begin{definition} \label{def:contingent_cone}
		(\textbf{Contingent cone}) Let $\mathcal{F}$ be a closed subset of $\mathcal{X} \subset \mathbb{R}^n$ and $x$ belongs to $\mathcal{F}$. The
		contingent cone to $\mathcal{F}$ at $x$ is the set
		\begin{align}
			T_\mathcal{F}(x) = \left\{v \in \mathcal{X} |\,\,\,\,  \liminf\limits_{h \rightarrow 0^+} \frac{d_\mathcal{F}(x +hv)}{h} =0  \right\}. 
		\end{align}
	\end{definition}
	It has been shown in \cite{blanchini1999set} that though the distance function $d_\mathcal{F}(y)$ depends on
	the considered norm, the set $T_\mathcal{F}(x)$ does not. Furthermore, the set $T_\mathcal{F}(x)$ is non-trivial (i.e., not identical to $\mathbb{R}^n$) only on the boundary of $\mathcal{F}$. 
	
	A key result in the set-invariance analysis, the celebrated Nagumo theorem, is stated as follows.
	\begin{theorem}  \label{Theorem_Nagumo}
		(\textbf{Nagumo theorem}, \cite{blanchini2008set}) Consider the system
		$\dot x (t) = f (x(t))$, and assume that, for each initial condition in
		a set $\mathcal{X} \subset \mathbb{R}^n$, it admits a globally unique solution. Let $\mathcal{F} \subset \mathcal{X}$ be
		a closed   set. Then the set $\mathcal{F}$ is positively
		invariant for the system if and only if
		\begin{align} \label{eq:Nagumo1}
			f (x(t)) \in T_\mathcal{F}(x), \,\,\, \forall x \in \mathcal{F}, 
		\end{align}
		where $T_\mathcal{F}(x)$ denotes the \textit{contingent cone} of $\mathcal{F}$ at $x$. 
	\end{theorem}
	
	
	If $x$ is an interior point in the set $\mathcal{F}$, then $T_\mathcal{F}(x) = \mathbb{R}^n$. Therefore, the condition in Theorem \ref{Theorem_Nagumo} is only meaningful when $x \in \text{bnd}(\mathcal{F})$, where $\text{bnd}(\mathcal{F})$ denotes the boundary of $\mathcal{F}$. Thus, the condition in \eqref{eq:Nagumo1} can be equivalently stated  as
	\begin{align}
		f (x(t)) \in T_\mathcal{F}(x), \,\,\, \forall x \in \text{bnd}(\mathcal{F}). 
	\end{align}
	The above condition clearly has an intuitive and geometric interpretation: if at $x \in \text{bnd}(\mathcal{F})$, the derivative
	$\dot x   = f(x(t))$ points inside or  is tangent to $\mathcal{F}$, then the
	trajectory $x(t)$ remains in $\mathcal{F}$.

	Consider a viable set $\mathcal{F}$ parameterized by a finite set of inequalities of the form 
	\begin{align} \label{eq:set_K}
		\mathcal{F} = \{x \in \mathcal{X} \,\,| \,\,g_i(x) \leq 0, i = 1, \cdots, m \},
	\end{align}
	where $g_i(x): \mathbb{R}^n \rightarrow \mathbb{R}$ are continuously differentiable functions \footnote{The condition that the $g_i(x)$ are continuously differentiable functions can be relaxed to requiring that the  $g_i(x)$ are Fr{\'e}chet differentiable at the point $x$ \cite{di1994contingent}. } defined in the set $\mathcal{X} \subset \mathbb{R}^n$.  Denote by $\chi(x)$ the set of \textit{active} constraints: $\chi(x) = \{i = 1, \cdots m \,\,| \,\,g_i(x) = 0\}$. Then the calculation of the contingent cone is simplified as $T_\mathcal{F}(x) = \mathbb{R}^n$ for $\chi(x) = \emptyset$ (i.e., at the interior point $x$ of the set $\mathcal{F}$), and 
	\begin{align} \label{viable_condition_inequality}
		T_\mathcal{F}(x) = \{ v \in \mathcal{X} | \left \langle v, \nabla g_i(x)\right \rangle \leq 0, \forall i \in \chi(x) \}.
	\end{align}
	For the   set  $\mathcal{F}$ parameterized as in \eqref{eq:set_K}, a consequence of Nagumo theorem is the following lemma for a controlled-invariant set. 
	\begin{lemma} \label{lemma:invariant_set_control}
		(\textbf{Set-invariance in control}, \cite{blanchini1999set,blanchini2008set})  Consider a set $\mathcal{F}$ parameterized as in \eqref{eq:set_K}. Then the set $\mathcal{F}$ is positively invariant under the dynamic control system $\dot x(t) = f(x(t), u(t))$ if $\dot x(t) \in \mathbb{R}^n$ for $\chi(x) = \emptyset$, and $\dot x(t) \in T_\mathcal{F}(x)$ with the contingent cone $T_\mathcal{F}(x)$ of \eqref{viable_condition_inequality}, or equivalently 
		\begin{align}
			\left< \nabla g_i(x),  f(x(t), u(t))  \right> \leq 0, \,\,\, \forall i \in \chi(x).
		\end{align}
	\end{lemma}

	\subsection{Problem formulation}
	
	Consider a group of $n$ vehicles, whose motion equations are described by the control-affine systems \eqref{eq:system_drift} with possibly different kinematics and/or drift terms. We assign the vehicle group with a coordination task, described by inter-vehicle geometric equality or inequality constraints that incorporate formation, flocking or other cooperative tasks. The four key problems on cooperative constrained motion coordination to be addressed in this paper are the following:
	\begin{itemize}
		\item (\textit{Mixed kinematics, equality and inequality constraints}) Model networked homogeneous or heterogeneous vehicle coordination tasks with both equality and inequality constraints;
		\item  (\textit{Coordination feasibility}) Determine whether a group of   heterogeneous vehicles can perform a coordination task that satisfies various motion constraints and coordination function constraints;
		\item (\textit{Constrained motion coordination}) If the coordination task with various constraints is feasible, determine feasible motions that generate trajectories for an $n$-vehicle group to realize a cooperative coordination  task; 
		\item (\textit{Time-varying coordination tasks}) Determine whether a group of   heterogeneous vehicles can perform a coordination task with various time-varying constraints; and if feasible, determine coordinated motions for all vehicles. 
	\end{itemize}

	\section{Formulation and transformation of time-invariant coordination constraints} \label{sec:constraints_formulation}
	\subsection{Motion constraints arising from   vehicle kinematics: affine (co)-distribution formulation} \label{sec:constraint_kinematics}
	In this subsection we follow the techniques in \cite{murray1994mathematical,Sun2016feasibility} to formulate each vehicle's kinematic motion constraints using (affine) distributions and codistributions. 
	A vehicle's motion equation modelled by a  nonlinear control-affine system \eqref{eq:system_drift} with drifts can be equivalently described by the following \textit{affine} distribution
	\begin{align} \label{eq:distribution}
		\Delta_i(p_i)  = f_{i,0}(p_i) + \text{span}\{f_{i,1}(p_i), f_{i,2}(p_i), \cdots, f_{i,l_i}(p_i)\}.
	\end{align}
	For the system \eqref{eq:system_drift} with drifts that is equivalently described by \eqref{eq:distribution}, one can obtain a corresponding transformation with equivalent motion  constraints via the construction of  a spanning set of covectors of the codistribution
	\begin{align} \label{eq:drift_equi}
		\omega_{i, j}(p_i) \dot p_i = q_{i, j}(p_i),  \quad j =1, \cdots, n_i - l_i,
	\end{align}
	where  the term $q_{i, j}$ is due to the existence of the drift term $f_{i,0}$.  We collect all the \textit{row covectors} $\omega_{i, j}$ to construct the codistribution matrix 
	\begin{align}
		\Omega_{K_i}(p_i) = [\omega_{i,1}^{\top}, \; \omega_{i,2}^{\top},\; \cdots,\; \omega_{i,n_i - l_i}^{\top}]^{\top},
	\end{align}
	and similarly define 
	\begin{align}
		T_{K_i}(p_i) = [q_{i,1}, q_{i,2}, \cdots, q_{i,n_i - l_i}]^{\top}.
	\end{align}
	Then one can rewrite \eqref{eq:drift_equi} in a compact form using affine codistributions as follows
	\begin{align} \label{eq:dynamics_constraint_transformation}
		\Omega_{K_i}(p_i)  \dot p_i = T_{K_i}(p_i),
	\end{align}
	where the subscript ``$K$'' stands for \emph{kinematics}.
	Furthermore, we collect all the kinematic motion constraints for all the $n$ vehicles in a composite form
	\begin{subequations}
		\begin{align}
			\Omega_K &= [\Omega_{K_1}^{\top}, \Omega_{K_2}^{\top}, \cdots, \Omega_{K_n}^{\top}]^{\top},  \\
			T_K &= [T_{K_1}^{\top}, T_{K_2}^{\top}, \cdots, T_{K_n}^{\top}]^{\top}.
		\end{align}
	\end{subequations}
	For ease of notation, we denote all of the vehicles' states    by the composite state vector $P = [p_1^{\top}, p_2^{\top}, \cdots, p_n^{\top}]^{\top}$. Thus, the overall kinematic motion constraint for all the vehicles can equivalent be stated compactly via its codistribution as $\Omega_K(P) \cdot \dot P = T_K(P)$.
	
	\begin{remark}
		The kinematic motion equation of the drift-free vehicle 
		\begin{align} \label{eq:system_drift_free}
			\dot p_i = \sum_{j=1}^{l_i} f_{i,j}(p_i) u_{i,j},
		\end{align}
		can be described in an equivalent form by the codistribution 
		\begin{align} \label{eq:drift_free_equi}
			\omega_{i, j}(p_i) \dot p_i = 0,  \,\,\,j =1, \cdots, n_i - l_i,
		\end{align}
		i.e.,  the term  $q_{i, j}$ in \eqref{eq:drift_equi} becomes zero. The above transformation is based on the idea that  a distribution generated by vector fields $\{f_{i,1}, \cdots, f_{i,l_i}\}$ of a nonlinear control system can be equivalently formulated by its  annihilating codistribution \cite{nijmeijer1990nonlinear}. Note that each $\omega_{i, j}(p_i)$ in \eqref{eq:drift_free_equi} is a \emph{row} covector in the dual space $(\mathbb{R}^{n_i})^\star$.
	\end{remark}
	
	\subsection{Motion constraints arising from time-invariant coordination} \label{sec:coordination_constraints}
	In this subsection we formulate motion constraints from coordination tasks using distributions/codistributions. We consider two types of constraints, \textit{equality} constraints and \textit{inequality} constraints, which both  involve inter-vehicle geometric relationships, in modelling a general form of coordination tasks. 
	\subsubsection{Networked vehicle interaction as an undirected graph}
	We assume a networked multi-vehicle control system modelled by an undirected graph $\mathcal{G}$, in which we use   $\mathcal{V}$ to denote its vertex set and $\mathcal{E}$ to denote the edge set.  The vertices consist of   $n$ homogeneous or heterogeneous vehicles, each modelled by the general dynamical equation \eqref{eq:system_drift} with possibly different motion constraints. The graph consists of  $m$ edges, each associated with one or multiple inter-vehicle constraints describing a coordination task.
	
	\subsubsection{Coordination with equality constraints}

	A family of equality constraints $\Phi$ is indexed by the edge set, denoted as $\Phi_\mathcal{E} = \{\Phi_{ij}(p_i, p_j)\}_{(i,j)}$ with  $ {(i,j)} \in \mathcal{E}$. For each edge $(i,j)$, $\Phi_{ij}(p_i, p_j)$ is a continuously differentiable vector function of the states $p_i$ and $p_j$  defining the coordination constraints between the vehicle pair $i$ and $j$.
	The constraint for edge $(i,j)$ is enforced if $\Phi_{ij} (p_i, p_j) = 0$.  Such equality constraints can be used to describe general coordinate control problems, such as formation shape control, distance maintenance, tracking and coverage control. For example, in formation shape control, the constraint vector function $\Phi_{ij}$ can be geometric functions of desired relative positions, or desired bearings, or desired distances between vehicles $i$ and $j$ in a target formation (see e.g., \cite{oh2015survey}). To satisfy the equality constraint for edge $(i,j)$, it should hold that
	\begin{align} \label{eq:constraints_ij}
		\frac{\text{d}}{\text{d}t} \Phi_{ij}(p_i, p_j) = \frac{\partial \Phi_{ij}}{\partial p_i} \dot p_i + \frac{\partial \Phi_{ij}}{\partial p_j} \dot p_j   = 0.
	\end{align}
	
	We collect the equality constraints for all the edges and define an overall  constraint denoted by $\Phi_\mathcal{E}(P)  = [\cdots, \Phi_{ij}^{\top}, \cdots]^{\top} = \bf 0$. A coordination task is  maintained if $\Phi_\mathcal{E} (P) = 0$ is enforced for all the edges.  Coordination feasibility with equality constraints means that the constraints are strictly satisfied along the trajectories of all vehicles in time, which implies the equivalent differential condition 
	\begin{align} \label{eq:constraints_all_1}
		\frac{\text{d}}{\text{d}t} \Phi(P) =\frac{\partial \Phi}{\partial P} \dot P   = 0,
	\end{align}
	
	By identifying a codistribution matrix $\Omega_E(P)$  associated with the Jacobian $\partial \Phi/\partial P$ in the nominal dual coordinate bases $[\text{d}P]$, 
	we   can reexpress Eq.~\eqref{eq:constraints_all_1} as
	\begin{equation}\label{eq:equality_constraint_overall1}
		\Omega_E(P) \dot P = 0.
	\end{equation}
	where the subscript ``$E$'' stands for \emph{equality} constraints.  Thus, the vector field $\dot P$ constrained by Eq. \eqref{eq:equality_constraint_overall1} represents possible motions for all the vehicles that respect the coordination equality constraint.
	
	\subsubsection{Coordination with inequality constraints}
	
	Now we consider a feasible coordination problem involving \textit{inequality} constraints. A family of inequality coordination constraints $\mathcal{I}_\mathcal{E} = \{\mathcal{I}_{ij}(p_i, p_j)\}_{(i,j)}$ is indexed by the edge set $\mathcal{E}$, and  each edge $(i, j)$ is associated with a vector function $\mathcal{I}_{ij}(p_i, p_j)$ which is assumed continuously differentiable. The constraints for the edge $(i, j)$ are enforced if $\mathcal{I}_{ij}(p_i, p_j) \leq 0,\;\;\forall t>0$. Now we consider the subset of \textit{active} constraints among all the edges
	\begin{align}
		\chi(P) = \{(i, j), \;i, j = 1,2, \cdots, n\; |\; \mathcal{I}_{ij}(p_i, p_j) = 0\}.
	\end{align}
	We remark that the set $\chi(P)$ is a dynamic set along time, which only involves the edge set with active constraints when the condition $\mathcal{I}_{ij}(p_i, p_j) \leq 0$ is about to be violated.   An inequality constraint for edge $(i,j)$ is maintained then if 
	\begin{align}
		\frac{\text{d}}{\text{d}t} \mathcal{I}_{ij} (p_i, p_j)= \frac{\partial \mathcal{I}_{ij}}{\partial p_i} \dot p_i + \frac{\partial \mathcal{I}_{ij}}{\partial p_j} \dot p_j \leq  0, \,\,\,\forall\,(i, j) \in \chi(P).
	\end{align}
	The active constraints in the edge set $\chi(P)$ generate a codistribution
	\begin{align}
		\Omega_I(P) = [\cdots, \Omega_{I, ij}^\top, \cdots]^\top, \,\,\,\forall (i, j) \in \chi(P),
	\end{align}
	where the   subscript ``$I$'' stands for \emph{inequality} constraints, and   $\Omega_{I, ij}(p_i, p_j)$ is obtained by the Jacobian of the vector function $\mathcal{I}_{ij}$ using the nominal coordinate bases $[\text{d}p_i, \text{d}p_j]$ associated with the active constraint $\mathcal{I}_{ij}(p_i, p_j) = 0$. 
	Based on the Nagumo theorem and Lemma~\ref{lemma:invariant_set_control},  to guarantee the validity of the inequality constraints, the control input $u(t) = [u_1(t)^\top, \cdots, u_n(t)^\top]^\top$ for each vehicle should be designed such that $\Omega_I(P) \dot P(P, u(t)) \leq 0$, $\forall (i, j) \in \chi(P)$. 
	
	\begin{remark}
		The expressions of the codistributions  $\Omega_E$ for equality constraints  and $\Omega_I$ of active inequality constraints are both coordinate-free and   independent of the enumeration of edge sets. However, one can always choose the nominal coordinate bases $[\text{d}P]$ to represent the codistributions $\Omega_E$, $\Omega_K$ and $\Omega_I$ in a matrix form. 
	\end{remark}
	
	\section{Coordination feasibility and motion generation with time-invariant constraints} \label{sec:coordination_feasibility}
	\subsection{Coordination feasibility of  multiple vehicles with inequality task constraints}
	
	We now state the following theorem on a feasible coordination for an $n$-vehicle group with kinematic constraints and inequality constraints in a coordination task. 
	\begin{theorem}\label{thm:feasible_equality}
		The coordination task for multi-vehicle systems described by \eqref{eq:system_drift} and with inequality coordination constraints has feasible motions if and only if the following mixed differential-algebraic equations with inequalities have solutions
		\begin{subequations}
			\begin{align}
				\Omega_K(P) \dot P &= T_K(P),  \label{eq:theorem1_motion_kine} \\  
				\Omega_I(P)\dot P  & \leq  0,\,\,\, \forall (i, j) \in \chi(P), \label{eq:theorem1_motion_coordination}
			\end{align}
		\end{subequations}
		where $\chi(P)$ denotes the set of active constraints among all the edges. 
	\end{theorem}
	
	\begin{proof}
		The transformation of the kinematic motion constraints in \eqref{eq:theorem1_motion_kine} is equivalent   to the distributions and co-distribution description as detailed in Section~\ref{sec:constraint_kinematics}.  The motion inequalities in \eqref{eq:theorem1_motion_coordination} are both   necessary and sufficient to ensure all coordination inequality tasks are satisfied for the overall vehicle group due to the Nagumo  Theorem~\ref{Theorem_Nagumo} and the controlled-invariant set condition of Lemma \ref{lemma:invariant_set_control}. 
	\end{proof}

	\subsection{Coordination feasibility of multiple vehicles with mixed equality and inequality task   constraints}
	We now consider   coordination tasks with both equality and inequality constraints.  
	Together with the active inequality constraints, one can state the following theorem that determines coordination feasibility with various constraints. 
	\begin{theorem} \label{thm:equality_inequality_1}
		The coordination task for multi-vehicle systems described by \eqref{eq:system_drift} and with both equality and inequality constraints has feasible motions if and only if the following \textit{mixed} differential-algebraic equations and inequalities have solutions
		\begin{subequations}
			\begin{align}
				\Omega_K(P) \dot P &= T_K(P),  \label{eq:theorem2_motion_kine} \\
				\Omega_E(P) \dot P &= 0,  \label{eq:theorem2_motion_coordination_eq} \\
				\Omega_I(P) \dot P  & \leq  0,\,\,\, \forall (i, j) \in \chi(P), \label{eq:theorem2_motion_coordination_ineq}
			\end{align}
		\end{subequations}
		where $\chi(P)$ denotes the set of active constraints among all the edges. 
	\end{theorem}
	
	\begin{proof}
		The transformation of the kinematic motion constraints in \eqref{eq:theorem2_motion_kine} is equivalent   to the distributions and co-distribution description of the motion equation \eqref{eq:system_drift} as detailed in Section~\ref{sec:constraint_kinematics}. 
		The equality condition in \eqref{eq:theorem2_motion_coordination_eq} is due to the special case of equality constraint maintenance for the controlled-invariant set of Lemma \ref{lemma:invariant_set_control}. 
		The motion inequalities in \eqref{eq:theorem2_motion_coordination_ineq} are both   necessary and sufficient to ensure all coordination inequality tasks are satisfied for the overall vehicle group due to the Nagumo  Theorem~\ref{Theorem_Nagumo}.
	\end{proof}
	
	
	\subsection{Generating vehicle's motion and trajectory for a feasible coordination} \label{sec:motion}
	The feasibility conditions presented in Theorems \ref{thm:feasible_equality} and~\ref{thm:equality_inequality_1} involve the determination of the existence of solutions for a  differential-algebraic equation (or a mixed inequality with motion equations). Solving these   equations with inequalities also leads to feasible motions that generate trajectories for each individual vehicle that meets both its own kinematic motion constraints and the inter-vehicle constraints for performing a coordination task. Generally speaking, when a solution exists that meets the differential-algebraic equations/inequalities, then such a solution is not unique. Any feasible trajectories can be generated by possible motions as described by the solutions of these equations/inequalities. 
	
	We remark that available approaches in numerical differential geometry and nonlinear control (see e.g., \cite{kwatny2000nonlinear}) are helpful and can be employed in solving these algebraic equations/inequalities. Furthermore, certain commercial software 
	(e.g., \textit{Matlab} or \textit{Mathematica}) has powerful toolboxes available that can perform symbolic computations if the number of symbolic variables is within a reasonable scale. 
	
	We propose Algorithm \ref{algorithm:feasibility_undirected} that presents a constructive approach to determine coordination feasibility and motion generation for the multi-vehicle coordination control under both equality and inequality constraints (Theorem \ref{thm:equality_inequality_1}). 
	

	\begin{algorithm}[t!]
		Initialization: $\Omega_{K_i}$, $T_{K_i}$, $\Omega_{E}$,   $\chi(P)$, $\Omega_I$;\;
		Construct the overall kinematic codistribution matrix~$\Omega_K$ and the vector $T_{K}$.\;
		
		\BlankLine
		\While{Running}{
			\BlankLine
			\textit{Solve equality} $ \left[ \begin{array}{c}
				\Omega_K\\
				\Omega_E   
			\end{array} \right]  \dot P=\left[ \begin{array}{c}
				T_K \\
				0   
			\end{array} \right]
			$\\
			\uIf{Solution does not exist}{
				Return: \textit{Coordination task is infeasible} (No solution);\;
				Condition checking STOP.\;
			}
			\Else{
				Calculate a special solution to the above equality constraint equation, denoted by $\bar K$;\;
				Determine $\kappa$ vectors that span $\text{Null}\left(\left[ \begin{array}{c}
					\Omega_K\\
					\Omega_E   
				\end{array} \right]\right)$, denoted by $K_1, K_2, \cdots, K_\kappa$.\;
			}
			\BlankLine
			\uIf{$\chi(P) = \emptyset$ (No active inequality constraint) }{
				Generate feasible motions $\dot P = \bar K + \sum_{l=1}^\kappa K_l  w_l$, where $w_l$ is a set of (constant or time-varying) virtual inputs that activate the associated vector field  $K_l$;\;
				Return: a set of feasible motions $\dot P = \bar K + \sum_{l=1}^\kappa K_l  w_l$ (according to different choices of $ w_l$).\;
			}
			\Else{
				Calculate and obtain a codistribution matrix $\Omega_I$ for active equality constraints, with $\forall (i, j) \in \chi(P)$;\;
				\For{$l=1, 2, \cdots, \kappa$}{
					Select a set of vectors $w_l$  \\
					\If{$\Omega_I (\bar K + \sum_{l=1}^\kappa K_l  w_l)    \leq  0$ for certain $ w_l$ }{ 
						Return: A feasible motion $\dot P =\bar K + \sum_{l=1}^\kappa K_l  w_l$. \;
					}
				}
				\If{$l > \kappa$}{Return: Feasible solution not found. Try again.}
			}
			
		}
		\caption{Coordination  feasibility checking and motion generation.}
		\label{algorithm:feasibility_undirected}
	\end{algorithm}
	
	Following Algorithm  \ref{algorithm:feasibility_undirected},  a  feasible motion for the overall heterogeneous networked system can be generated by the following equivalent dynamical system
	\begin{align} \label{eq:abstract_motion}
		\dot P = \bar K + \sum_{l=1}^\kappa K_l  w_l
	\end{align}
	where $w_l$ is an input that activates the associated vector field  $K_l$.
	When a feasible motion is determined with a set of virtual input $w_l$, the actual control input $u_i = [u_{i,1}, \cdots, u_{i,l_i}]^\top$ can be readily calculated via each vehicle's kinematic equations as follows. 
	We construct the compact form of the motion equations of all vehicles
	\begin{align}
		\dot P = \underbrace{\left[
			\begin{array}{c}
				f_{1,0}\\
				f_{2,0}\\
				\vdots \\
				f_{i,0} \\
				\vdots \\
				f_{n,0}
			\end{array}
			\right]}_{F_0} + \underbrace{
			\begin{bmatrix}
				F_{1} &  &  &  &  & \\ 
				& F_{2} &  &  &  & \\ 
				&  &  \ddots &  &  & \\ 
				&  &   & F_{i} &  & \\
				&  &   &   & \ddots &  & \\
				&  &   &   &   &  F_{n}
		\end{bmatrix}}_{F} \underbrace{
			\left[
			\begin{array}{c}
				u_{1}\\
				u_{2}\\
				\vdots \\
				u_{i} \\
				\vdots \\
				u_{n}
			\end{array}
			\right]}_{u}
	\end{align}
	where $F_i = [f_{i,1}, f_{i,2}, \cdots, f_{i,l_i}]$, and $F$ is a block diagonal matrix that collects all vector field matrix functions $F_i$ of each vehicle $i$. By the condition that the distribution matrix $F_i$ has a full column rank (i.e., the vehicle states $p_i$ for the spanned distribution matrix $F_i(p_i)$ are regular points), the block diagonal matrix $F$ is also of full column rank. Therefore, given a set of virtual inputs $w_l$ that generates feasible coordinated motions for all vehicles in Equation~\eqref{eq:abstract_motion}, the following equation has a unique solution in terms of the actual control input $u$,
	\begin{align}
		\dot P &=  F_{0}+ Fu = \bar K + \sum_{l=1}^\kappa K_l  w_l  
	\end{align}  
	and the control input for the heterogeneous vehicle group can be derived by 
	\begin{align}
		u &= (F^\top F)^{-1}F^\top \left(\bar K + \sum_{l=1}^\kappa K_l  w_l - F_{0} \right)
	\end{align}

	\section{Coordination of unicycles with multiple vehicles under inter-vehicle geometric constraints} \label{sec:case_study}
	In this section, we consider several application examples with case studies to illustrate the proposed coordination theory and algorithms. These application examples involve the coordination of homogeneous or heterogeneous vehicles subject to various combinations of constraints. Since unicycle vehicle control is a benchmark problem in robotic motion generation, we will in particular discuss  unicycle coordination control with multiple vehicles.

	\subsection{Typical vehicle  models and their distribution/codistribution formulation of   kinematic motion constraints} \label{sec:kinematic_model}
	We consider four types of heterogeneous vehicles:  unicycle-type vehicle with nonholonomic motion constraints,  constant-speed vehicle with both nonholonomic motion and speed constraints,  car-like vehicle with under-actuated motions, and fully-actuated vehicle modeled by integrators.
	
	\subsubsection{Nonholonomic unicycle  vehicle}
	The unicycle vehicle is described by \begin{align} \label{eq:example_unicycle}
		\dot x_i &= v_i \,\, \text{cos}(\theta_i), \nonumber \\
		\dot y_i &= v_i \,\, \text{sin}(\theta_i),   \\
		\dot \theta_i &=   u_i,  \nonumber
	\end{align}
	where the state variable is $p_i = [x_i, y_i, \theta_i]^{\top} \in  \mathbb{R}^3$. The unicycle vehicle in the model equation \eqref{eq:example_unicycle} should satisfy ``rolling without slipping", making
	it  a nonholonomic vehicle. 
	In an equivalent and  compact form, we can write
	\begin{align}
		\dot p_i = [\dot x_i, \dot y_i, \dot \theta_i]^{\top} = f_{i,1} v_i + f_{i,2} u_i,
	\end{align}
	with
	\begin{align} \label{eq:vector_field_vehicle2}
		f_{i,1} = \left[
		\begin{array}{ccc}
			\text{cos}(\theta_i)\\
			\text{sin}(\theta_i)\\
			0
		\end{array}
		\right],
		f_{i,2} = \left[
		\begin{array}{ccc}
			0\\
			0\\
			1
		\end{array}
		\right].
	\end{align}
	
	The kinematic constraint for a nonholonomic unicycle  vehicle is described by the distribution $\Delta_i = \text{span}\{f_{i,1}, f_{i,2}\}$, which can be equivalently stated by the annihilating codistribution
	\begin{align}
		\Omega_{K_i} = \Delta_i^\perp = \text{span}\{\text{sin}(\theta_i) \text{d}x_i-\text{cos}(\theta_i) \text{d}y_i\},
	\end{align}
	which has a full rank for all values of $\theta_i$. 
	\subsubsection{Unicycle-type vehicle with constant speed}
	Now consider a nonholonomic unicycle-type vehicle with \textit{constant-speed constraints}. The motion equations of such vehicles can   be described by \eqref{eq:example_unicycle}, where the  speed term $v_i$ is \textit{fixed} and
	the only control input is $u_i$ that steers the vehicle's orientations. Such equations have been commonly used to model the planar kinematics of fixed-wing UAVs with speed constraints \cite{beard2012small, sun2021collaborative}. 
	By introducing the two vector fields  
	\begin{align} \label{eq:vector_field_vehicle2}
		f_{i,0} = \left[
		\begin{array}{ccc}
			v_i\text{cos}(\theta_i)\\
			v_i\text{sin}(\theta_i)\\
			0
		\end{array}
		\right],
		f_{i,1} = \left[
		\begin{array}{ccc}
			0\\
			0\\
			1
		\end{array}
		\right],
	\end{align}
	we can rewrite the constant-speed vehicle model as a control affine system with a drift term
	\begin{equation} \label{eq:constant_speed_model}
		\dot p_i = [\dot x_i, \dot y_i, \dot \theta_i]^{\top} = f_{i,0} + f_{i,1} u_{i}.
	\end{equation}
	Denote the  two linearly independent covectors  of   the codistribution   as $\omega_{i,1}$ and $\omega_{i,2}$. 
	With the dual vector basis $(\text{d}x_i, \text{d}y_i, \text{d}\theta_i)$, one can write an explicit expression  for a pair of linearly independent covectors  (see \cite{Sun2016feasibility})
	\begin{align}
		\omega_{i,1} &=  \text{sin}(\theta_i) \text{d}x_i - \text{cos}(\theta_i)\text{d}y_i, \nonumber \\
		\omega_{i,2} &=  \text{cos}(\theta_i)\text{d}x_i+\text{sin}(\theta_i)\text{d}y_i . \nonumber
	\end{align}
	The affine codistribution is obtained as $\Omega_{K,i} = [
	\omega_{i,1}^{\top},
	\omega_{i,2}^{\top}
	]^{\top}$, and there holds   $\Omega_{K,i}f_{i,1} = 0$ and $\Omega_{K,i}f_{i,0} = T_{K_i}$, where $T_{K_i} = [q_{i,1}, q_{i,2}]^{\top} = [0 , v_i]^{\top}$. 
	
	\subsubsection{Car-like vehicle}
	Further consider a car-like vehicle, whose kinematic equation is described by (see \cite{de1998feedback})
	\begin{align} \label{eq:car_model}
		\dot x_i &= u_{i,1}   \text{cos}(\theta_i), \nonumber \\
		\dot y_i &= u_{i,1}  \text{sin}(\theta_i),   \nonumber \\
		\dot \theta_i &= u_{i,1} (1/l_i)  \text{tan} (\phi_i),    \nonumber \\
		\dot \phi_i &=  u_{i,2},
	\end{align}
	with the state variables $p_i = (x_i, y_i, \theta_i, \phi_i) \in   \mathbb{R}^4$, where $(x_i, y_i)$ are the Cartesian coordinates of the rear wheel, $\theta_i$ is the orientation angle of the vehicle
	body with respect to the $x$ axis,  $\phi_i \in (-\pi/2, \pi/2)$ is the steering angle, and $l_i$ is the distance between the midpoints of the two wheels. The control inputs $u_{i,1}$ and $u_{i,2}$ are the \textit{driving} and the \textit{steering} velocity inputs, respectively.
	The model \eqref{eq:car_model} describes kinematic motions for a typical rear-wheel-driving car, which is subject to two nonholonomic motion constraints (rolling without slipping sideways for each wheel, respectively). 
	In an equivalent compact form, one can write
	\begin{align}
		\dot p_i = [\dot x_i, \dot y_i, \dot \theta_i, \dot \phi_i]^{\top} = f_{i,1} u_{i,1} + f_{i,2} u_{i,2},
	\end{align}
	with 
	\begin{align}
		f_{i,1} = 
		\left[
		\begin{array}{c}
			\text{cos}(\theta_i)\\
			\text{sin}(\theta_i)\\
			(1/l_i)  \text{tan} (\phi_i)\\
			0
		\end{array}
		\right],
		f_{i,2} = \left[
		\begin{array}{c}
			0 \\
			0 \\
			0 \\
			1
		\end{array}
		\right].
	\end{align}
	The kinematic constraints  by the two vector fields $f_{i,1}$ and $f_{i,2}$ are described by the distribution $\Delta_i = \text{span}\{f_{i,1}, f_{i,2}\}$. After some algebraic calculations, one can   equivalently state the kinematic constraints of the car-like vehicle by the annihilating codistribution:
	\begin{align}
		\Omega_{K_i} &= \Delta_i^\perp \nonumber \\
		&=  [\text{sin}(\theta_i+ \phi_i) \text{d}x_i-\text{cos}(\theta_i+\phi_i) \text{d}y_i - l_i \text{cos}(\phi_i)\text{d}\theta_i, \nonumber \\
		&\,\,\,\,\,\,\,\,\,\,\text{sin}(\theta_i) \text{d}x_i-\text{cos}(\theta_i) \text{d}y_i],
	\end{align}
	which has a full and constant rank for all relevant values of the variables. 
	\begin{figure}[t]
		\begin{center}
			\vspace{-10pt}
			\includegraphics[width=0.4\textwidth]{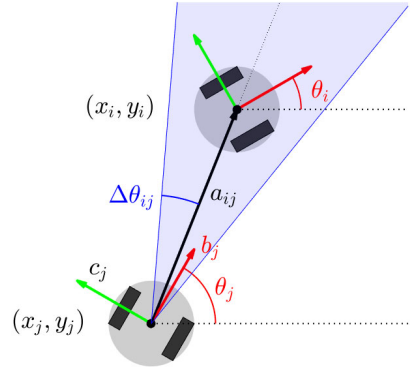}
			\vspace{-20pt}
			\caption{Illustration of a visibility inequality constraint, $\mathcal{I}^{(4)}_{ij}$, bounding the direction $b_j$ to the blue cone defined by $a_{ij}$ and the angle $\Delta\theta_{ij}$.}
			\label{fig:geometry_1}
		\end{center}
	\end{figure}
	
	\subsubsection{Fully-actuated vehicle modeled by integrators}
	We also consider fully-actuated vehicles whose dynamical equation is 
	\begin{align} \label{eq:integrator}
		\dot x_i = u_{i,1}, \,\,\,\,
		\dot y_i = u_{i,2},
	\end{align}
	where $u_{i,1}, u_{i,2}$ are the control inputs which directly generate corresponding velocities in the $x$ and $y$ coordinates. For such fully-actuated vehicles,  the two vectors $f_{i,1} = [1, 0]^{\top}$ and $f_{i,2} =[0, 1]^{\top}$ span $\mathbb{R}^2$ ($\Delta_i = \mathbb{R}^2$) and thus the codistribution is null ($\Omega_{K_i} = \emptyset$).
	\subsection{Typical geometric coordination  constraints: distance, heading and visibility functions} \label{sec:modelling_constraints}
	Consider two of the previously defined vehicles in the form~\eqref{eq:system_drift} sub-indexed $i$ and $j$, respectively, as illustrated in Figure~\ref{fig:geometry_1}. A common coordination task may include a simple inter-vehicle distance constraint, with 
	\begin{align}\label{eq:distanceconstraint}
		\Phi_{ij}^{d}:  \frac{1}{2}(x_i - x_j)^2 +  \frac{1}{2}(y_i - y_j)^2  -  \frac{1}{2}d_{ij}^2 = 0,
	\end{align}
	for some $d_{ij}>0$, which generates a codistribution matrix
	\begin{equation}
		\Omega_{E, ij}^{d} = [(x_i - x_j)(\text{d}x_i - \text{d}x_j) + (y_i - y_j)(\text{d}y_i - \text{d}y_j)].
	\end{equation}
	Practical coordination tasks may also include a distance constraint in terms of a two-sided inequality,
	\begin{small}
		\begin{align}  \label{eq:distance_inequality_constraint}
			\mathcal{I}_{ij}^{d}:  \frac{1}{2}(d_{ij}^{\mathrm{min}})^2\leq  \frac{1}{2}(x_i - x_j)^2 +  \frac{1}{2}(y_i - y_j)^2 \leq  \frac{1}{2}(d_{ij}^{\mathrm{max}})^2,  
		\end{align}
	\end{small}
	with $d_{ij}^{\mathrm{max}}>d_{ij}^{\mathrm{min}}>0$, and codistribution matrix given by $\Omega_{I, ij}^{d}=\Omega_{E, ij}^{d}$ if the right inequality becomes active, or   $\Omega_{I, ij}^{d}= -\Omega_{E, ij}^{d}$ if the left inequality becomes active. 
	
	Other tasks may require heading constraints in the form
	\begin{align} \label{eq:heading_constraint_linear}
		\Phi_{ij}^{h}: \theta_i - \theta_j = \delta_{ij},
	\end{align}
	for some constant $\delta_{ij}>0$. The corresponding codistribution of this constraint takes the form $\Omega_{E, ij}^{h} = [\theta_i \text{d} \theta_i - \theta_j \text{d}\theta_j]$.  A more general and relaxed form of the heading constraint is given by
	\begin{align} \label{eq:heading_constraint3}
		\mathcal{I}_{ij}^{h}: \Delta\theta_{ij}^{\mathrm{min}}\leq\text{arctan}\Big(\frac{y_i-y_j}{x_i - x_j}\Big) - \theta_j \leq \Delta\theta_{ij}^{\mathrm{max}},
	\end{align}
	referred to as a visibility constraint. Such an inequality constraint  has  been used in modelling visibility maintenance control in multi-robotic systems \cite{panagou2014cooperative, liu2018vision}. However, the inequality heading constraint in the form of \eqref{eq:heading_constraint3}  is restrictive due to the range of the arctangent function. Consequently, we consider an equivalent inequality constraint. Assume that $\Delta\theta_{ij}=-\Delta\theta_{ij}^{\mathrm{min}}=\Delta\theta_{ij}^{\mathrm{max}}$, and define the directions $a_{ij} := [x_i - x_j, y_i - y_j]$, $b_{j} := [\cos(\theta_j), \sin(\theta_j)]$, $c_{j} := [-\sin(\theta_j), \cos(\theta_j)]$. Then we can write the heading constraint in Equation~\eqref{eq:heading_constraint3} as
	\begin{align} \label{eq:heading_inequality_constraint3}
		\mathcal{I}_{ij}^{a}: \cos(\Delta\theta_{ij})\langle a_{ij}, a_{ij} \rangle^{1/2}\leq  \langle a_{ij},b_{j} \rangle.
	\end{align}
	Here, it can be verified that the associated codistribution matrix is
	\begin{align}\label{eq:headingconst}
		\begin{array}{c}
			\Omega_{I,ij}^{a}\hspace{-1pt}=\hspace{-1pt}\dfrac{\langle a_{ij}, c_j\rangle}{\sqrt{\langle a_{ij}, a_{ij} \rangle}} \Bigg(
			\dfrac{1}{\langle a_{ij}, a_{ij} \rangle} \Bigg\langle a_{ij},
			\hspace{-3pt}
			\Bigg[
			\begin{array}{c}
				\hspace{-5pt}\text{d}x_i-\text{d}x_j\hspace{-5pt}\\
				\hspace{-5pt}\text{d}y_j-\text{d}y_i\hspace{-5pt}
			\end{array}
			\Bigg]
			\hspace{-1pt}
			\Bigg\rangle
			\hspace{-2pt}+
			\hspace{-2pt}\text{d}\theta_j\Bigg),
		\end{array}
	\end{align}
	when the inequality constraint \eqref{eq:heading_inequality_constraint3} becomes active. 
	\begin{remark}\label{rem:singA}
		It should be noted that the constraint~\eqref{eq:headingconst} may become singular when vehicles $i$ and $j$ collide (i.e., $a_{ij} = 0$), due to the division by $\langle a_{ij}, a_{ij} \rangle$. This case can always be excluded by satisfying the distance constraints in \eqref{eq:distanceconstraint} or  \eqref{eq:distance_inequality_constraint}. Therefore, the generated co-distribution matrices that model the geometric coordination tasks are always well-defined and have constant ranks under the geometric inter-vehicle constraints.  
	\end{remark}
	
	\subsection{Coordinating two nonholonomic unicycle vehicles}\label{sec:twounicycles}
	
	In the first example, we consider two unicycle vehicles which are to cooperatively maintain a constant inter-vehicle distance \eqref{eq:distanceconstraint} and a bounded heading displacement or visibility inequality constraint as described above by one of  \eqref{eq:heading_constraint3} or \eqref{eq:heading_inequality_constraint3}. 
	Now we construct a joint codistribution matrix from the nonholonomic kinematic motion constraints and the distance equality constraint
	\begin{align}
		\Omega =  
		\left[
		\begin{array}{cccccc}
			\text{sin}(\theta_1)& - \text{cos}(\theta_1)& 0 &0 &0 &0   \\
			0 &0 &0 & \text{sin}(\theta_2) & -\text{cos}(\theta_2)&  0   \\
			x_1 - x_2 &y_1 - y_2 & 0     &x_2 - x_1 &y_2 - y_1 &  0
		\end{array} \right] 
	\end{align}
	with $T =  [T_K^\top, T_E^\top]^\top = [0, 0, 0]^\top$. Solving the equations $\Omega(P)\dot P = T$ yields the feasible motion solutions $\dot P = \sum_{i=1}^3 w_iK_i$, where
	\begin{align}  \label{eq:solution_unicycle}
		K_1  &=
		\left[
		0,
		0,
		1,
		0,
		0,
		0   
		\right]^\top,   
		K_2   =
		\left[
		0,
		0,
		0,
		0,
		0,
		1   
		\right]^\top \nonumber \\
		K_3  &=
		\left[
		\begin{array}{c}
			\text{cos}(\theta_1)\left(\text{cos}(\theta_2)(x_1 - x_2) + \text{sin}(\theta_2) (y_1 - y_2) \right)   \\
			\text{sin}(\theta_1)\left(\text{cos}(\theta_2)(x_1 - x_2) + \text{sin}(\theta_2) (y_1 - y_2) \right)   \\
			0   \\
			\text{cos}(\theta_2)\left(\text{cos}(\theta_1)(x_1 - x_2) + \text{sin}(\theta_1) (y_1 - y_2) \right)   \\
			\text{sin}(\theta_2)\left(\text{cos}(\theta_1)(x_1 - x_2) + \text{sin}(\theta_1) (y_1 - y_2) \right)   \\
			0  
		\end{array} \right]
	\end{align}
	It is clear that the virtual controls $w_1$ and $w_2$ generate the angular speeds for each vehicle, respectively, while the term $K_3$ maintains a constant desired distance between them (assuming that initially the distance constraint is met). Furthermore, the solution with $K_1$ and $K_2$ and virtual control inputs $w_1$ and $w_2$ possesses the motion freedoms to generate an admissible angular input that achieves desired heading re-orientations to satisfy the heading or visibility inequality in the form of  \eqref{eq:heading_constraint_linear}-\eqref{eq:heading_inequality_constraint3}. For example, when the heading inequality constraint becomes active in the sense that $\theta_1 - \theta_2 -\delta_{12}^+ = 0$ which renders a codistribution $\Omega_{I, 12}^{(2)}=[\theta_i \text{d} \theta_i - \theta_j \text{d}\theta_j]$, any $w_1K_1$ with a negative $w_1$, or any $w_2K_2$ with a positive $w_2$, is a feasible solution guaranteeing $\Omega_{I, 12}^{(2)} \dot P \leq 0$ that generates feasible motions for the vehicle group.  The same principle is also applied to other types of heading inequality constraints in the form of \eqref{eq:heading_constraint3} or \eqref{eq:heading_inequality_constraint3}, while a feasible motion always exists to ensure that the heading or visibility inequality constraint is always satisfied. In summary, we have the following lemma on coordination feasibility and motion generation for a two-unicycle vehicle group.
	
	\begin{lemma} \label{lemma:two_vehicle_exm1}
		Consider two unicycle-type vehicles, each described by \eqref{eq:example_unicycle},   with a coordination task of maintaining a constant inter-vehicle distance   $d_{12}$ and a bounded heading displacement or visibility inequality constraint. Suppose initially at  time $t = 0$ both constraints are met. By using the above derived control   solutions with the vector functions $K_1, K_2, K_3$:
		\begin{itemize}
			\item The coordination task with inter-vehicle distance constraint is preserved by the derived control vector fields with any $w_1, w_2, w_3$.
			\item If initially the heading/visibility  inequality is satisfied, then a feasible control always exists (with the possible use of non-zero $w_l$) that preserves both distance equality and heading/visibility inequality constraints. 
		\end{itemize}
	\end{lemma}
	
	\begin{remark}
		Following the approach proposed in Section~\ref{sec:motion} and based on the derived abstraction motion   equations $\dot P = w_1 K_1 + w_2 K_2 + w_3 K_3$ with the motion vector fields in \eqref{eq:solution_unicycle}, we can recover the actual control input ($v_i, u_i, i = 1,2$) for the nonholonomic unicycle vehicles of \eqref{eq:example_unicycle} as follows
		\begin{align}
			\text{Vehicle 1}: 
			\begin{cases}
				v_1\,\,= &w_3\left({\cos}(\theta_2)(x_1 - x_2) + {\sin}(\theta_2) (y_1 - y_2) \right) \nonumber \\
				u_1\,\,= &w_1 \nonumber \\
			\end{cases} \\
			\text{Vehicle 2}:   
			\begin{cases}
				v_2\,\,=& w_3\left({\cos}(\theta_1)(x_1 - x_2) + {\sin}(\theta_1) (y_1 - y_2) \right) \nonumber \\
				u_2\,\,=& w_2
			\end{cases}
		\end{align}
		The recent papers  \cite{hajieghrary2017cooperative, hajieghrary2018differential} have discussed cooperative transport of a buoyant load via ASVs modelled by unicycle equations, while the transport task is constrained by  a constant distance between two unicycle vehicles. The authors of \cite{hajieghrary2017cooperative, hajieghrary2018differential} have studied  a differential geometric approach   to derive feasible cooperative motions.  We note that using the proposed method in this paper, more constraints (including distance and headings) can be incorporated in the cooperative motion design for unicycle coordination. 
	\end{remark}


	\subsection{Coordinating a unicycle and a constant-speed vehicle} \label{sec:constspeedunicycle}
	Now we consider a coordination task that involves a constant-speed vehicle and a general unicycle vehicle, aiming to maintain inter-vehicle distance equality and heading inequality constraints for a coordination task. This is a typical coordination control problem incorporating a fixed-wing UAV (with speed constraints) and ground mobile objects such as ground vehicles \cite{wang2014vision,wang2019coordinated}. To solve this problem, we find the codistribution matrix from the kinematic motion equations and coordination equality constraint, constructed by 
	\begin{align}
		\Omega =
		\left[
		\begin{array}{cccccc}
			\text{sin}(\theta_1)& - \text{cos}(\theta_1)& 0 &0 &0 &0   \\
			\text{cos}(\theta_1)&  \text{sin}(\theta_1) &0 &0 &0 &0 \\
			0 &0 &0 & \text{sin}(\theta_2) & -\text{cos}(\theta_2)&  0    \\
			x_1 - x_2 &y_1 - y_2 & 0     &x_2 - x_1 &y_2 - y_1 &  0
		\end{array} \right]
	\end{align}
	with $T =  [T_K^\top, T_E^\top]^\top = [0, v_1,0,0,0]^\top$. The algebraic equation $\Omega(P)\dot P = T$ is solved by
	\begin{align}
		\bar K = 
		\left[
		\begin{array}{c}
			v_1 \text{cos}(\theta_1)   \\
			v_1 \text{sin}(\theta_1)   \\
			0   \\
			\frac{\text{cos}(\theta_2)\left(v_1 \text{cos}(\theta_1)(x_1 - x_2) + v_1 \text{sin}(\theta_1)(y_1 - y_2)\right)}{\text{cos}(\theta_2)(x_1 - x_2) + \text{sin}(\theta_2)(y_1 - y_2)}    \\
			\frac{\text{sin}(\theta_2)\left(v_1 \text{cos}(\theta_1)(x_1 - x_2) + v_1 \text{sin}(\theta_1)(y_1 - y_2)\right)}{\text{cos}(\theta_2)(x_1 - x_2) + \text{sin}(\theta_2)(y_1 - y_2)}    \\
			0   
		\end{array} \right]
	\end{align}
	and $
	K_1  =
	\left[
	0,
	0,
	1,
	0,
	0,
	0   
	\right]^\top,
	K_2  =
	\left[
	0,
	0,
	0,
	0,
	0,
	1    \right]^\top
	$, which enables an abstraction of the coordination system of feasible motion equations $\dot P = \bar K + \sum_{l=1}^2 w_l K_l$. An analysis analogous to Lemma \ref{lemma:two_vehicle_exm1} delivers  the following result.
	\begin{lemma}
		Consider a unicycle-type vehicle and a constant-speed vehicle in a coordination group to maintain inter-vehicle distance equality and heading inequality or visibility constraints described in Section~\ref{sec:modelling_constraints}. By using the above-derived control solutions $\bar K, K_1, K_2$: 
		\begin{itemize}
			\item The coordination task with inter-vehicle distance is preserved with the derived  control vector fields for any $w_1$ and $w_2$. 
			\item If initially the heading or visibility inequality is satisfied, then a feasible motion always exists  with possible use of non-zero $w_1$ or $w_2$ that preserves both distance equality and heading/visibility inequality constraints. 
		\end{itemize}
	\end{lemma}

	\subsection{Coordinating a nonholonomic unicycle and a car-like vehicle}
	Consider a two-vehicle group, one described by the unicycle equation and the other by a car-like dynamics. The two vehicles assume a task to cooperatively maintain a constant distance $d_{12}$ and a heading or visibility inequality constraint.
	
	The joint codistribution matrix from both the kinematic constraints and the distance equality constraint can be obtained as (using the standard bases for the dual space $[\text{d}x_1,\text{d}x_2, \cdots, \text{d}\phi_2,\text{d}\theta_2]$): 
	\begin{align}
		\Omega = [&\text{sin}(\theta_1) \text{d}x_1-\text{cos}(\theta_1) \text{d}y_1, \nonumber \\
		&\text{sin}(\theta_2+ \phi_2) \text{d}x_2-\text{cos}(\theta_2+\phi_2) \text{d}y_2 - l_2 \text{cos}(\phi_2)\text{d}\theta_2, \nonumber \\
		&\text{sin}(\theta_2) \text{d}x_2-\text{cos}(\theta_2) \text{d}y_2,  \nonumber \\
		&(x_1 - x_2)(\text{d}x_1 - \text{d}x_2) + (y_1 - y_2)(\text{d}y_1 - \text{d}y_2)].
	\end{align}
	The solution to the algebraic equation $\Omega(P)\dot P = T = 0$ is obtained as
	$
	\dot P = \sum_{l=1}^3 w_l K_l
	$ with 
	$
	K_1  =
	\left[
	0,
	0,
	1,
	0,
	0,
	0,
	0   
	\right]^\top,
	K_2  =
	\left[
	0,
	0,
	0,
	0,
	0,
	0,
	1   
	\right]^\top  
	$, and 
	\begin{align}
		K_3  =
		\left[
		\begin{array}{c}
			\text{cos}(\theta_1)\left(\text{cos}(\theta_2)(x_1 - x_2) + \text{sin}(\theta_2) (y_1 - y_2) \right)   \\
			\text{sin}(\theta_1)\left(\text{cos}(\theta_2)(x_1 - x_2) + \text{sin}(\theta_2) (y_1 - y_2) \right)   \\
			0   \\
			\text{cos}(\theta_2)\left(\text{cos}(\theta_1)(x_1 - x_2) + \text{sin}(\theta_1) (y_1 - y_2) \right)   \\
			\text{sin}(\theta_2)\left(\text{cos}(\theta_1)(x_1 - x_2) + \text{sin}(\theta_1) (y_1 - y_2) \right)   \\
			\frac{1}{l_2} \text{tan}\phi_2 \left(\text{cos}(\theta_1)(x_1 - x_2) + \text{sin}(\theta_1) (y_1 - y_2) \right)\\
			0  
		\end{array} \right]
	\end{align}
	The coordination feasibility and motion generation result is summarized in the following lemma.
	\begin{lemma}\label{lem:visconst}
		Consider a two-vehicle group consisting of a  unicycle-type vehicle and a car-like vehicle, with a coordination task of maintaining a constant inter-vehicle distance   $d_{12}$ and a   heading/visibility constraint described in Section~\ref{sec:modelling_constraints}.  By using the above derived control solutions $K_1, K_2, K_3$: 
		\begin{itemize}
			\item The inter-vehicle distance is preserved with the above derived control vector fields for any choice of $w_l$. 
			\item If initially the heading/visibility inequality is satisfied, then a feasible control always exists that preserves both distance equality and heading/visibility inequality constraints. 
		\end{itemize}
	\end{lemma}

	\subsection{Simulation study: complex heterogeneous coordination}\label{sec:complexleaderfollower}
	To show the versatility and demonstrate the veracity of the theory, we give a final example consisting of two unicycle vehicles with states $p_i = [x_i,y_i,\theta_i]^{\top}$, $i\in\{1,2\}$ obeying the kinematics in~\eqref{eq:example_unicycle}, and a single integrator model $[x_3,y_3]^{\top}$ obeying the kinematics in~\eqref{eq:integrator}. The matrix $\Omega_K\in \mathbb{R}^{2\times 8}$ is formed with $T_K = [0, 0]^{\top}$ as in Section~\ref{sec:twounicycles}, and one unicycle vehicle is constrained to deliver a desired reference velocity. We consider this reference in terms of the smooth velocities  $v_{r}(t), u_{r}(t)$, and enforce it through the equality constraints given by the affine codistributions,
	\begin{equation}
		\cos(\theta_1)\text{d}x_1 + \sin(\theta_1)\text{d}y_1 = v_{r}(t), \quad \text{d}\theta_1 = u_{r}(t),
	\end{equation}
	which may be represented in a compact matrix form with $\Omega_E\in \mathbb{R}^{2\times 8}$ with $T_E = [v_{r}(t), u_{r}(t)]^{\top}\in \mathbb{R}^2$. In this case and for the simulation purpose, we let
	$v_{r}(t) = 2\sin(\tfrac{5t}{4}), \quad u_{r}(t) = 3\cos(\tfrac{7t}{4})$. 
	\begin{figure} 
		\centering
		\includegraphics[width=0.8\columnwidth]{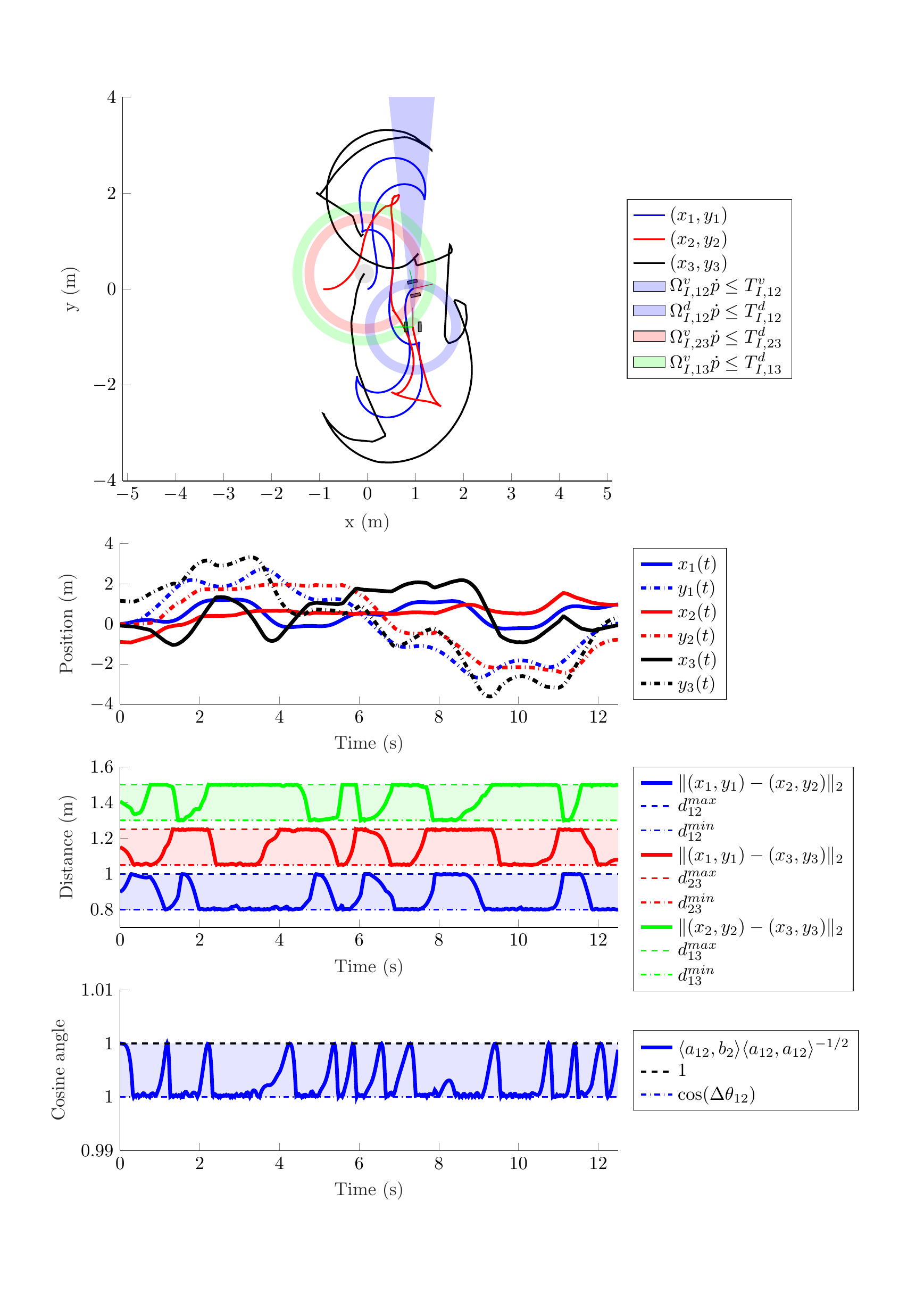} 
		\caption{Complex coordination example. The first vehicle ($i = 1$, in blue) is set to follow a predefined velocity profile, the motions of the second and third vehicles ($i = 2$ in red and $i =3$ in black) are constrained such that distance inequality constraints are satisfied between all vehicles (depicted as rings). Furthermore, a visibility constraint is added for the second vehicle (depicted as a cone). \textit{Top:} The system configurations and constraints at $t = 12.5$. \textit{Top, center}: The positional trajectories of the system. \textit{Bottom, center}: The distance inequality constraints, and the actual inter-vehicle distances in the system response. \textit{Bottom}: The visibility constraint and the corresponding visibility functions in the system response. A video of the simulation can be accessed through~\cite{greiff2021video}.} \label{fig:complexleaderfollower}
	\end{figure}
	In addition to these equality constraints, we form several inequality constraints. Vehicle $i=2$ is to maintain visibility of the vehicle $i=1$, and we pose a distance inequality constraint $\mathcal{I}_{12}^{d}$ with very narrow bounds, $d_{12}^{\mathrm{min}}=0.8$ and $d_{12}^{\mathrm{max}}=1.0$ (m), and a visibility constraint with a very tight angle bound $\Delta\theta_{12}=0.1$ (rad). Combined, the constraints define  an extremely narrow feasible region, illustrated as the intersection of the blue regions in Fig~\ref{fig:complexleaderfollower}. In addition to these constraints, we seek to enforce distance inequality constraints between all of the vehicles. To this end, we define two additional inequality constraints enforcing lower and upper bounds on the distance $\|(x_j,y_j)-(x_3,y_3)\|$ for $j \in\{1,2\}$, characterized by  the bounds $d_{13}^{\mathrm{min}} = 1.05$, $d_{13}^{\mathrm{max}} = 1.25$, $d_{23}^{\mathrm{min}} = 1.3$, and $d_{23}^{\mathrm{max}} = 1.5$. In total, this results in three kinematic equality constraints, two  equality constraints for reference trajectories, and seven time-invariant inequality constraints. By the same reasoning as in Lemma~\ref{lem:visconst}, we expect any initially feasible solution to remain feasible, as is observed in Figure~\ref{fig:complexleaderfollower}. Despite the fast and erratic movement of first vehicle ($i=1$), the constraints are satisfied at all times, and this vehicle remains visible to vehicle $i=2$ (it always resides in the union of the blue regions), and all of the distance inequality constraints are satisfied at all times. {\footnote{Several simulation videos on constrained motion coordination of multiple homogeneous or heterogeneous vehicles can be found in \url{https://youtu.be/wFx6x7Mep74}.}}
	
	In this particular example, there exist many feasible motions, and the control signals are selected so as to minimize a slew-rate constraint on the system velocities $\min_{w_l}\|(\mathrm{d}^2/\mathrm{d} t^2){P}(t)\|_{2}^2$. In other words, we seek to keep system velocities  as constant as possible when not switching vector fields. The intuition is that this can help with following moving vehicles with minimal switching, but we emphasize many other criteria for the selection of $w_l$ can be considered for numerical implementations.

	\section{Heterogeneous multi-vehicle coordination with time-varying constraints} \label{sec:time_varying_coordination}
	In this section, we develop a general framework for treating constrained motion coordination control with time-varying constraints, including mixed equality and inequality task constraints. First, a general  temporal viability theory will be developed, followed by controlled temporal invariance for control affine systems. Then, we present the feasibility theorem and motion generation for multiple   heterogeneous vehicles with mixed time-varying coordination constraints. 
	
	\subsection{Temporal viability theory for time-varying constraints} \label{sec:temporal_viability}
	In this subsection \footnote{Parts of results in Section \ref{sec:temporal_viability}-\ref{sec:temporal_viability_regulation} are presented in \cite{greiff2019temporal}; here we present full details and proofs of the main results. } we present several general concepts on temporal viability with time-varying constraints, and develop novel results on controlled temporal invariance for dynamical control systems. 
	
	Consider the following time-varying control dynamical system described by a general ordinary differential equation
	\begin{align} \label{eq:control_system}
		\dot x(t) = f(x(t), u(t), t),
	\end{align}
	where $x \in \mathbb{R}^n$ is the state variable, $u \in \mathbb{R}^l$ is the control input vector, and $f \in \mathbb{R}^n$ is a (possibly time-varying) vector field of the state $x(t)$, control input $u(t)$ and the time $t$. 
	
	Following the conventional definition of viability theory and viable set \cite{aubin2009viability},  we define \textit{temporal} viability and \textit{time-varying} viable set. 
	
	\begin{definition}
		(\textbf{Temporal viability and time-varying viable set}) Consider a control system described by a differential equation $\dot x(t) = f(x(t), t)$  in \eqref{eq:control_system}. A subset $\mathcal{F}(t) \subset \mathbb{R}^n$ enjoys the \textit{temporal} viability property for the system \eqref{eq:control_system} under the time interval $t \in [\tilde t, \bar t)$ if for every initial state $x(\tilde t) \in \mathcal{F}(\tilde t)$ at time~$\tilde t$, there exists at least one solution to the system
		starting at $x(\tilde t)$ which is viable in the time interval $[\tilde t, \bar t)$ in the sense that 
		\begin{align}
			\forall t \in [\tilde t, \bar t), \,\,\, x(t) \in \mathcal{F}(t).
		\end{align}
		
		The set $\mathcal{F}(t)$ is then termed a \textit{time-varying} viable set for the dynamical control system \eqref{eq:control_system} over the time interval $[\tilde t, \bar t)$. 
	\end{definition}
	
	In the following, without loss of generality we will assume the initial time $\tilde t = 0$. Further, the time $\bar t$ then denotes the maximum existence time until which the  solution of the dynamical system  \eqref{eq:control_system} can be extended.  If the solution of the dynamical system \eqref{eq:control_system} can be extended to infinity, we may also consider all the positive times, i.e., $\bar t \rightarrow \infty$. 
	{\footnote{When a differential equation that models a dynamical control system involves discontinuous right-hand side (e.g., switching controls), we understand its solutions in the sense of Filippov \cite{cortes2008discontinuous}. In the case of discontinuous dynamic control systems, the theorem on controlled set invariance is generalized to weak invariance, and the developed main results still apply \cite{cortes2008discontinuous}.  }}
	
	We  define a temporal distance function for a point $y$ to a   time-varying set $\mathcal{F}(t)$ as $d^t(y, \mathcal{F}(t)) := \inf\limits_{z \in \mathcal{F}(t)} \|y-z(t)\|$ at time $t$.  As in \cite{aubin2009set} and \cite{hauswirth2018time}, we define the temporal contingent cone as follows.
	
	\begin{definition} \label{def:temporal_contingent}
		(\textbf{Temporal contingent cone}) Let $\mathcal{F}(\hat t)$ be a closed subset of $\mathcal{X} \subset \mathbb{R}^n$ and $x(t)$ belongs to $\mathcal{F}(t)$ at time $t = \hat t$. The temporal
		contingent cone \footnote{We remark that, as has been shown in \cite{hauswirth2018time}, the temporal contingent cone is not necessarily a cone. Similarly to \cite{hauswirth2018time}, 
			we follow the convention and term it `temporal contingent cone' as it reduces to the standard contingent cone in Definition \ref{def:contingent_cone} when the set $\mathcal{F}$ is time independent.} to $\mathcal{F}(t)$ at $x(t)$ and time $\hat t$ is the set
		\begin{align}
			T_\mathcal{F}^{\hat t}(x) = \left\{v \in \mathcal{X} |\,\,\,\,  \liminf\limits_{h \rightarrow 0^+} \frac{d^{\hat t}\left(x (\hat t)+hv, \mathcal{F}(\hat t +h)\right)}{h} =0  \right\}. 
		\end{align}
	\end{definition}
	
	

	Now consider a   time-varying  set $\mathcal{F}$ parameterized by a set of time-varying function inequalities of the form
	\begin{align} \label{eq:parameterize_set_F}
		\mathcal{F}(t) = \{(x(t),t) | \,\,| \,\,g_i(x(t), t) \leq 0, i = 1, \cdots, m\},
	\end{align}
	where $g_i(x(t), t): \mathbb{R}^n \times \mathbb{R} \rightarrow \mathbb{R}$. 
	We impose the following assumption on each time-varying function $g_i(x(t), t)$ for deriving a well-defined temporal contingent cone with favourable properties. 
	\begin{assumption} \label{assum:function_g}
		The function   $g_i(x(t), t): \mathbb{R}^n \times \mathbb{R} \rightarrow \mathbb{R}$ is a $C^1$  function of the state variable $x(t)$, and is Lipschitz continuous with respect to  the time variable~$t$.  
	\end{assumption}
	
	Denote by $\chi(x, t)$ the time-varying set of \textit{active} constraints at time $t$: $\chi(x, t) = \{(i, t) \,\,| \,\,g_i(x, t) = 0\}$.
	The following result characterizes an explicit formula of temporal contingent cone $T_\mathcal{F}^{t}(x)$ when the set $\mathcal{F}(t)$ is parameterized by  a set of time varying inequalities  as in \eqref{eq:parameterize_set_F} under Assumption \ref{assum:function_g}.
	\begin{lemma} \label{lemma_temporal_cone}
		Consider the time-varying set $\mathcal{F}(t)$ parameterized by a set of time-varying inequalities with functions $g_i(x(t), t)$ as in~\eqref{eq:parameterize_set_F}. Assume that the Jacobian matrix  $\nabla_x^\top g_\chi(x(t), t)$ is of full rank for the points $(x(t), t)$ with the stacked active constraint function vector $g_\chi(x(t), t) = [\cdots, g_i(x(t), t)\cdots]^\top, i \in\chi(x(t), t)$. Then the temporal contingent cone is described by 
		\begin{align} \label{eq:inactive_set_cone}
			T_\mathcal{F}^{t}(x) = \mathbb{R}^n, \,\,\,\,\text{if}\,\,\chi(x, t) = \emptyset
		\end{align}
		
		and 
		\begin{align} \label{eq:active_set_cone2}
			T_\mathcal{F}^{t}(x) = \left\{v \middle| \begin{bmatrix}
				(\nabla_x^\top g_\chi(x(t), t)),
				\nabla_t g_\chi(x(t), t)
			\end{bmatrix}  
			\begin{bmatrix}
				v \\
				1
			\end{bmatrix} \leq 0 \right\}, \nonumber \\
			\,\,\,\,\text{if}\,\,\chi(x, t) \neq \emptyset.
		\end{align}
		
	\end{lemma}
	
	A detailed calculation of the contingent cone and the proof follows similarly to that of \cite{hauswirth2018time} and is omitted here. 

	\begin{remark}
		A non-empty temporal contingent cone $T_\mathcal{F}^{t}(x)$   for all time $t$ is a necessary condition to ensure the existence of the control input $u(t)$ associated with the time-varying vector field $f(x(t), u(t), t)$. A sufficient condition to guarantee non-empty $T_\mathcal{F}^{t}(x)$ along the solution $x(t)$ and time $t$ is the \textit{forward Lipschitz continuity} of the set $\mathcal{F}(t)$ with respect to time $t$ (see \cite[Theorem 1]{hauswirth2018time}). According to \cite[Proposition 4]{hauswirth2018time}, a sufficient condition to ensure the forward Lipschitz continuity of the set $\mathcal{F}(t)$ is (i) the Jacobian matrix  $\nabla_x^\top g_\chi(x(t), t)$ has full rank, and (ii) the time-varying function $g_i(x(t, t))$ is Lipschitz continuous in $t$. In this paper we may occasionally further suppose $g_i(x(t), t)$ is a $C^1$ function of both the state $x(t)$ and time $t$ whose Jacobian has full rank,   which automatically guarantees the second condition. By imposing Assumption \ref{assum:function_g}, in the following we always ensure  that the set $\mathcal{F}(x(t),t)$ parameterized by a set of  inequality/equality constraints of the time-varying function $g_i(x(t), t)$ is forward Lipschitz continuous. 
	\end{remark}

	\subsection{Controlled temporal invariance}  \label{sec:temporal_viability_regulation}
	We first present the following theorem on controlled temporal invariance with time-varying constraints, which  extends the Nagumo theorem and  standard results in controlled invariance theory (\cite{blanchini1999set}).
	\begin{theorem} \label{theorem:invariant_temporal_set}
		(\textbf{Controlled temporal invariant set})  Consider a forward Lipschitz continuous time-varying set $\mathcal{F}(t)$ parameterized as in \eqref{eq:parameterize_set_F}. Then the set $\mathcal{F}(t)$ is positively temporal invariant under the dynamical control system $\dot x(t)$, $t \in [0, \bar t)$, of \eqref{eq:control_system} if $x(0) \in \mathcal{F}(0)$, and $\dot x(t) =  f(x(t), u(t), t) \in T_\mathcal{F}^{t}(x)$ with the temporal contingent cone $T_\mathcal{F}^{t}(x)$ defined in \eqref{eq:inactive_set_cone} and \eqref{eq:active_set_cone2} of Lemma \ref{lemma_temporal_cone}. 
		Equivalently, to guarantee the controlled temporal invariance of the set $\mathcal{F}(t)$, the (possibly time-varying) vector function $f$ should satisfy 
		\begin{align}
			f(x(t), u(t), t) \in \mathbb{R}^n, \,\,\, \text{if}\,\,\chi(x, t) = \emptyset
		\end{align}
		and
		\begin{align} \label{eq:control_input_general}
			\nabla_x^\top g_i(x(t), t) f(x(t), u(t), t) + \nabla_t g_i(x(t), t) \leq 0, \nonumber \\ \,\,\, \forall i \in \chi(x, t), \forall t \in [0, \bar t). 
		\end{align}
	\end{theorem}
	
	\begin{proof}
		The proof is based on the explicit formulas in Lemma~\ref{lemma_temporal_cone} that characterize the  temporal contingent cone set as a function  of  time. Forward Lipschitz continuity of the set $\mathcal{F}(t)$, which is guaranteed by Assumption \ref{assum:function_g} on the constraint function  $g_i(x(t), t)$ and full rank of $\nabla_x g_\chi(x(t), t)$ as computed from the active constraint functions $g_\chi$ at all time $t$, implies the temporal contingent cone is a non-empty set:  
		\begin{align}
			T_\mathcal{F}^{t}(x) \neq \emptyset, \forall t \in [0, \bar t). 
		\end{align}
		The necessary and sufficient condition to ensure that the time-varying set $\mathcal{F}(t)$ is viable under the time-varying control system $\dot x(t) =  f(x(t), u(t), t)$ is 
		\begin{align} \label{eq:necessary_suffi_condition}
			f(x(t), u(t), t) \cap T_\mathcal{F}^{t}(x)  \neq \emptyset, \forall t \in [0, \bar t). 
		\end{align}
		When the time-varying constraint functions are all inactive at time $t$ (i.e., $\chi(x, t) = \emptyset$ at time $t$), the temporal contingent cone is the whole space  $T_\mathcal{F}^{t}(x) = \mathbb{R}^n$ at time $t$, which implies that the time-varying vector function can be any vector $f(x(t), u(t), t) \in \mathbb{R}^n$. 
		When a set of constraint functions become active at time $t$, the temporal contingent cone formula in \eqref{eq:active_set_cone2} renders an equivalent formulation as in  \eqref{eq:control_input_general} to ensure that the viability condition of  \eqref{eq:necessary_suffi_condition} is always satisfied.  The control input \eqref{eq:control_input_general} provides  corrective actions that regulate the states of the dynamical system $\dot x(t) =  f(x(t), u(t), t) \in T_\mathcal{F}^{t}(x)$ to be controlled temporal  invariant in the set $\mathcal{F}(t)$.   
	\end{proof}

	When specializing the temporal viability theory to  control affine systems  
	of the   general form 
	\begin{align} \label{eq:control_affine_system}
		\dot p(t) = f_0(p(t)) + \sum_{j = 1}^l f_j(p(t)) u_j(t),
	\end{align}
	one obtains the following theorem on controlled temporal viability.
	\begin{lemma} \label{theorem:invariant_temporal_set_affine}
		(\textbf{Temporal viability of control affine systems})  
		\begin{itemize}
			\item (\textbf{Time-varying inequality constraints})
			Consider a forward Lipschitz continuous time-varying set $\mathcal{F}(t)$ parameterized  
			by a set of time-varying function inequalities of the form $
			\mathcal{F}(t) = \{(p(t),t) | \,\,| \,\,g_i(p(t), t) \leq 0, i = 1, \cdots, m\}
			$. 
			The set $\mathcal{F}(t)$ is controlled temporal viable under the control affine system $\dot p(t)$, $t \in [0, \bar t)$, of \eqref{eq:control_affine_system} if $p(0) \in \mathcal{F}(0)$ and the control input $u_j$ satisfies  (whenever a set of constraints is active):
			\begin{align} \label{eq:viable_condition_affine}
				& \sum_{j = 1}^l u_j   \nabla_p^\top g_i(p(t), t)  f_j(p(t))  \nonumber \\
				& \leq  -\nabla_p^\top g_i(p(t), t) f_0(p(t))    -\nabla_t g_i(p(t), t),  \nonumber \\ &\forall i \in \chi(x, t), \forall t \in [0, \bar t). 
			\end{align}
			and when all inequality constraints are inactive the control input takes arbitrary value in the sense that $u_j \in \mathbb{R}, j = 1, \cdots, l$. 
			
			\item (\textbf{Time-varying equality constraints}) Consider a forward Lipschitz continuous time-varying set $\mathcal{F}(t)$ parameterized by time-varying equality constraint of  $\mathcal{F}(t) = \{(p(t), t) | g(p(t), t) = 0\}$. The set $\mathcal{F}(t)$ is controlled temporal viable under the control affine system $\dot p(t)$ of \eqref{eq:control_affine_system} if $p(0) \in \mathcal{F}(0)$ and the control input $u_j$ satisfies    
			\begin{align} \label{eq:viable_condition_affine_equality}
				&\sum_{j = 1}^l u_j \nabla_p^\top g(p(t), t) f_j(p(t))   \nonumber \\ 
				&=  -\nabla_p^\top g(p(t), t) f_0(p(t)) -\nabla_t g(p(t), t).
			\end{align}
		\end{itemize}
		
	\end{lemma}
	
	The above lemma can be obtained as a consequence of Theorem \ref{theorem:invariant_temporal_set}, which also recovers the main result on controlled invariance of time-varying algebraic sets  in a recent paper \cite{yuno2020invariance}. 
	With the preparation of temporal viability theory, we are now ready to develop the coordination control framework with time-varying constraints. 
	
	
	
	\subsection{Coordination under time-varying constraint tasks}
	\label{sec:time_varying_coordination_constraints}
	
	\subsubsection{Coordination with time-varying equality constraints}
	In the coordination task, a set of time-varying equality constraints $\Phi$ is associated with the edge set of the coordination graph, denoted as $\Phi_\mathcal{E} = \{\Phi_{ij}(p_i, p_j, t)\}_{(i,j)}$ with  $ {(i,j)} \in \mathcal{E}$. For each edge $(i,j)$, $\Phi_{ij}$ is a time-varying continuously differentiable vector function of the states $p_i$ and $p_j$  defining time-varying coordination constraints between the vehicle pair $i$ and $j$.
	The constraint for edge $(i,j)$ is enforced if $\Phi_{ij} (p_i, p_j, t) = 0$.    To satisfy the equality constraint for edge $(i,j)$,  it should hold that
	\begin{align} \label{eq:constraints_ij}
		\frac{\text{d}}{\text{d}t} \Phi_{ij} = \frac{\partial \Phi_{ij}}{\partial p_i} \dot p_i + \frac{\partial \Phi_{ij}}{\partial p_j} \dot p_j +  \frac{\partial \Phi_{ij}}{\partial t}   = 0.
	\end{align}
	
	By collecting the equality constraints for all the edges and defining an overall  constraint denoted by $\Phi_\mathcal{E}  = [\cdots, \Phi_{ij}^{\top}, \cdots]^{\top} = \bf 0$, the coordination feasibility with equality constraints means that the constraints $\Phi_\mathcal{E}(P, t)$ are strictly satisfied along the  trajectories of all vehicles in time by the following
	\begin{align} \label{eq:constraints_all}
		\frac{\text{d}}{\text{d}t} \Phi =\frac{\partial \Phi}{\partial P} \dot P +  \frac{\partial \Phi}{\partial t} = 0,
	\end{align}
	
	Now we group all the constraints for all the edges by writing down a compact form $ T_E = - [\cdots, (\frac{\partial \Phi_{ij}}{\partial t})^{\top}, \cdots]^{\top}$ and identify a codistribution matrix $\Omega_E$  associated with the Jacobian $\partial \Phi/\partial P$ using the nominal dual coordinate bases $[\text{d}P]$. 
	We now can reexpress Equation \eqref{eq:constraints_all} as
	\begin{equation}\label{eq:equality_constraint_overall}
		\Omega_E(P)\dot P = T_E.
	\end{equation}
	The vector field  $\dot P$ defined by the above equation \eqref{eq:equality_constraint_overall} represents possible motions for all the vehicles that respect the coordination time-varying equality constraint.
	
	\subsubsection{Coordination with time-varying inequality constraints}
	A set of time-varying inequality coordination constraints $\mathcal{I}_\mathcal{E} = \{\mathcal{I}_{ij}\}_{(i,j)}$ is indexed by the edge set $\mathcal{E}$, and  each edge $(i, j)$ is associated with a time-varying vector function $\mathcal{I}_{ij}(p_i, p_j, t)$ that satisfies Assumption~\ref{assum:function_g}. The constraints for the edge $(i, j)$ are enforced if $\mathcal{I}_{ij}(p_i, p_j, t) \leq 0,\;\;\forall t>0$. Now we consider the subset of \textit{active} constraints among all the edges at time $t$
	\begin{align}
		\chi(P, t) = \{(i, j), \;i, j = 1,2, \cdots, n\; |\; \mathcal{I}_{ij}(p_i, p_j, t) = 0\}.
	\end{align}
	At any point in time, all the active constraints in the edge set $\chi(P)$ generate a codistribution
	\begin{align}
		\Omega_I = [\cdots, \Omega_{I, ij}^\top, \cdots]^\top, \,\,\,\forall (i, j) \in \chi(P, t),
	\end{align}
	where the   subscript $I$ stands for \emph{inequality} constraints, and   $\Omega_{I, ij}$ is obtained by the Jacobian of the vector function $\mathcal{I}_{ij}$ using the nominal coordinate bases $[\text{d}p_i, \text{d}p_j]$ associated with the active constraint $\mathcal{I}_{ij}(p_i, p_j, t) = 0$. 
	We group all the active constraints for all the edges and specify a compact form $ T_I = - [\cdots, (\frac{\partial \Omega_{I, ij}}{\partial t})^{\top}, \cdots]^{\top}$
	According to Lemma~\ref{theorem:invariant_temporal_set},  to guarantee that the time-varying inequality constraints are met, the control input $u(t) = [u_1(t)^\top, \cdots, u_n(t)^\top]^\top$ for each vehicle should be designed such that $\Omega_I \dot P(P, u(t)) \leq T_I$, $\forall (i, j) \in \chi(P, t)$.

	\subsection{Coordination feasibility and motion generation with time-varying  equality and inequality task   constraints}
	We now consider   coordination tasks with both equality and inequality constraints.  
	Together with the active inequality constraints, one can state the following theorem that determines coordination feasibility with various time-varying constraints. 
	\begin{theorem} \label{thm:time-varying_equality_inequality}
		The coordination task with both time-varying equality and inequality constraints has feasible motions if the following \textit{mixed} differential-algebraic equations and inequalities have solutions
		\begin{subequations}
			\begin{align}
				\Omega_K(P) \dot P &= T_K(P),  \label{eq:theorem5_motion_kine} \\
				\Omega_E(P) \dot P &= T_E,  \label{eq:theorem5_motion_eq} \\
				\Omega_I(P) \dot P  & \leq  T_I,\,\,\, \forall (i, j) \in \chi(P, t), \label{eq:theorem5_motion_ineq}
			\end{align} 
		\end{subequations}
		where $\chi(P, t)$ denotes the set of active constraints among all the edges at time $t$. 
	\end{theorem}
	\begin{proof}
		The transformation of the kinematic motion constraints in \eqref{eq:theorem5_motion_kine} is equivalent   to the distributions and co-distribution description as detailed in Section~\ref{sec:constraint_kinematics}. 
		The equality condition to maintain time-varying equality coordination tasks in \eqref{eq:theorem5_motion_eq} is due to the special case of equality constraint maintenance  (Lemma \ref{theorem:invariant_temporal_set_affine}). 
		The motion inequalities in \eqref{eq:theorem5_motion_ineq} are both   necessary and sufficient to ensure all coordination time-varying inequality tasks are satisfied for the overall vehicle group due to controlled temporal invariant set of Theorem~\ref{theorem:invariant_temporal_set} and control affine systems in   Lemma~\ref{theorem:invariant_temporal_set_affine}. In summary, the composite equality/inequality conditions in  \eqref{eq:theorem5_motion_kine}-\eqref{eq:theorem5_motion_ineq} are both necessary and sufficient to ensure that feasible motions for all vehicles satisfy both kinematic constraints and time-varying coordination tasks. 
	\end{proof}
	For motion generation, we can follow Algorithm~\ref{algorithm:feasibility_undirected} to compute a feasible motion for each heterogeneous vehicle. In this case, the first step of solving the codistribution  equality is updated by 
	\begin{align}
		\left[ \begin{array}{c}
			\Omega_K\\
			\Omega_E   
		\end{array} \right]  \dot P=\left[ \begin{array}{c}
			T_K \\
			T_E   
		\end{array} \right],
	\end{align}
	while the null vector in the composite codistribution matrix $ \Omega := [\Omega_K^\top,
	\Omega_E^\top]^\top$ remains unchanged, and the special solution $\bar K$ should be updated according to the temporal information $T = [T_K^\top,
	T_E^\top ]^\top $ in the time-varying constraints. Then with a set of virtual inputs $w_l$, one can further proceed to choose feasible motions to satisfy the codistribution \eqref{eq:theorem5_motion_ineq} of time-varying inequality constraints.

	\subsection{Revisiting a robot coordination example from \cite{tabuada2005motion}:  fully heterogeneous vehicles with time-varying equality constraints}
	\label{sec:applications}

	In this subsection, by using an illustrative example with a group of heterogeneous vehicles, we show how to apply the previous results to determine coordination feasibility and generate feasible motions under  time-varying coordination constraints.
	
	The illustrative example   is adopted from  \cite{tabuada2005motion}, but here we consider a more complicated multi-vehicle coordination problem when all vehicles have heterogeneous dynamics. Consider an undirected formation consisting of three vehicles, while  vehicle 1 is modelled by the kinematic equation \eqref{eq:constant_speed_model}  with both nonholonomic constraint and drift term, vehicle 2 is a unicycle-type robot with nonholonomic constraint and  adjustable forward speed described by  \eqref{eq:example_unicycle}, and vehicle 3 is modelled by single integrators \eqref{eq:integrator}. 
	
	By following the computation of (affine) codistributions of each vehicle kinematics in Section \ref{sec:kinematic_model} and putting together kinematic constraints for all vehicles, one can obtain a compact form of the co-distribution matrix 
	\begin{align}
		\Omega_K =  \left[
		\begin{array}{c}
			\text{sin}(\theta_1) \text{d}x_1 - \text{cos}(\theta_1)\text{d}y_1  \\
			\text{cos}(\theta_1)\text{d}x_1+\text{sin}(\theta_1)\text{d}y_1 \\
			\text{sin}(\theta_2) \text{d}x_2-\text{cos}(\theta_2) \text{d}y_2
		\end{array}
		\right]
	\end{align}
	with $T_K = [0, v_1, 0]^{\top}$.

	The coordination  constraint follows similarly from \cite[Example 4.2]{tabuada2005motion}, with some slight modifications. We suppose that the coordination task consists of maintaining  a relative position constraint and a time-varying distance constraint. Specifically, the constraint associated with the edge $(1,2)$ is defined by $ c_{12} = \left[
	x_1 - x_2 - \delta_x,
	y_1 - y_2 - \delta_y,
	\theta_1-\theta_2
	\right]^\top $,
	where $\delta_x$ and $\delta_y$ are positive constants specifying the desired relative displacement between vehicle 1 and vehicle 2.
	The constraint between vehicle 1 and vehicle 3 is defined by the function
	$
	c_{13}(t) = \frac{1}{2}(x_1 - x_3)^2 + \frac{1}{2}(y_1 - y_3)^2 - \delta_{13}(t)
	$,
	where $\delta_{13}(t)$ is a positive function specifying the desired time-varying distance between vehicle 1 and vehicle 3.
	The codistribution for characterizing the coordination constraint can be described by
	\begin{align} \label{eq:feasibility_codistribution_exampe2}
		\Omega_E =
		\left[
		\begin{array}{c}
			\text{d}x_1-\text{d}x_2 \\
			\text{d}y_1-\text{d}y_2 \\
			\text{d}\theta_1-\text{d}\theta_2 \\
			(x_1 - x_3) \text{d}x_1 + (y_1 - y_3) \text{d}y_1 +  (x_3 - x_1) \text{d}x_3 \\
			+  (y_3 - y_1) \text{d}y_3
		\end{array}
		\right]
	\end{align}
	and $T_E = [0,0,0, \gamma_{13}(t)]$, where $\gamma_{13}(t) = - \frac{\partial c_{13}(t)}{\partial t} = \frac{\partial \delta_{13}(t)}{\partial t}$, and we have used the standard basis $\{\text{d}x_1, \text{d}y_1, \text{d}\theta_1, \cdots, \text{d}x_3, \text{d}y_3\}$ to describe each covector in $\Omega_E$. By grouping both the kinematic constraint and formation constraint, one can write down a composite  constraint matrix $\Omega$ as shown in \eqref{eq:example} (see next page).

	\begin{figure*}[ht]
		\begin{subequations}
			\begin{align} \label{eq:example}
				\Omega &=   \begin{bNiceMatrix} 
					\Omega_K\\
					\hdottedline
					\Omega_E   
				\end{bNiceMatrix}   = 
				\begin{bNiceMatrix}
					{\text{sin}(\theta_1)}& - \text{cos}(\theta_1)& 0 &0 &0 &0 &0 &0  \\
					\text{cos}(\theta_1)&  \text{sin}(\theta_1) &0 &0 &0 &0 &0 &0 \\
					0 &0 &  {0} & \text{sin}(\theta_2) & -\text{cos}(\theta_2)&  0 &0 &0   \\
					\hdottedline
					1 &0 & 0  &-1 &0 &0 &0 &0    \\
					0 &1 & 0  &0 &-1 &0 &0 &0    \\
					0 &0 & 1  &0 &0 &-1 &0 &0   \\
					x_1 - x_3 &y_1 - y_3 & 0  &  0 &0 &0 &x_3 - x_1 &y_3 - y_1
				\end{bNiceMatrix},\hspace{20pt} 
				T =   \begin{bNiceMatrix} 
					T_K\\
					\hdottedline
					T_E   
				\end{bNiceMatrix}=
				\begin{bNiceMatrix}
					0   \\ v_1  \\ 0 \\
					\hdottedline 0 \\ 0 \\ 0 \\ \gamma_{13}(t)  
				\end{bNiceMatrix}
				\\
				\dot P &=
				\begin{bmatrix}
					\dot x_1 \\ \dot y_1 \\ \dot \theta_1 \\ \dot x_2 \\ \dot y_2 \\ \dot \theta_2 \\  \dot x_3 \\ \dot y_3
				\end{bmatrix},\hspace{20pt} 
				\bar K =
				\begin{bmatrix}
					v_1\text{cos}(\theta_1) \\  v_1 \text{sin}(\theta_1) \\ 0 \\ v_1\text{cos}(\theta_1) \\  v_1 \text{sin}(\theta_1) \\ 0 \\ v_1\text{cos}(\theta_1) \\  \frac{\gamma_{13}(t)}{y_3 - y_1} +v_1 \text{sin}(\theta_1)
				\end{bmatrix},\hspace{20pt} 
				K_1 =
				\begin{bmatrix}
					0 \\ 0 \\  0 \\ 0 \\  0 \\  0 \\ y_3 - y_1 \\  x_1 - x_3
				\end{bmatrix},\hspace{20pt} 
				K_2 =
				\begin{bmatrix}
					0 \\  0 \\ 1 \\ 0 \\  0 \\ 1 \\  0 \\  0
				\end{bmatrix}\label{eq:exampleB}
			\end{align}
		\end{subequations}
	\end{figure*}
	
	Straightforward calculation shows that  $\text{rank}(\Omega) = \text{rank}([\Omega, T])$ holds in this example. Thus, according to Theorem~\ref{thm:time-varying_equality_inequality}, there exist  feasible motions to coordinate such a heterogeneous vehicle group while maintaining the coordination task constraints. One special solution $\bar K$ to the equation $\Omega(P) \dot P =  T$ and the two null vectors $K_1$, $K_2$ of $\Omega$ are calculated in \eqref{eq:exampleB} (see next page).
	Then all feasible motions for the heterogeneous vehicle group can be described by the following equivalent abstraction  system
	\begin{align} \label{eq:motion_generation}
		\dot P =  \bar K +  w_1 K_1  +   w_2 K_2
	\end{align}
	Choosing different virtual inputs $w_1$ and $w_2$ will generate different coordinated motions for all vehicles while satisfying the time-varying coordination tasks.  More precisely, all feasible coordination motions can be categorized into the following four types:
	\begin{itemize}
		\item Type (i): the case of $w_1 = 0$ and $w_2 = 0$: vehicle 1 and vehicle 2  translate with the same heading while maintaining the formation constraint  $c_{12}(P) = 0$, and vehicle 3's motion is generated by $\dot x_3 = v_1\text{cos}(\theta_1), \dot y_3 = \frac{\gamma_{13}(t)}{y_3 - y_1} + v_1 \text{sin}(\theta_1)$ that satisfies the time-varying distance constraint $c_{13}(P, t) = 0$;
		\item Type (ii): the case of $w_1 \neq 0$ and $w_2 = 0$:  vehicle 1 and vehicle 2 translate with the same heading, while vehicle 3 performs a combined  motion of both translation and rotation generated by the dynamics $\dot x_3 = v_1\text{cos}(\theta_1) + w_1(y_3 - y_1), \dot y_3 = \frac{\gamma_{13}(t)}{y_3 - y_1} + v_1 \text{sin}(\theta_1) + w_1(x_3 - x_1)$;
		\item Type (iii): the case of $w_1 = 0$ and $w_2 \neq 0$: vehicle 1 and vehicle 2 rotate with the same angular velocity $w_2$ and the same instantaneous phase while maintaining the formation constraint  $c_{12}(P) = 0$, and vehicle 3's motion is generated by $\dot x_3 = v_1\text{cos}(\theta_1), \dot y_3 = \frac{\gamma_{13}(t)}{y_3 - y_1} + v_1 \text{sin}(\theta_1)$;
		\item Type (iv): the case of $w_1 \neq 0$ and $w_2 \neq 0$:  vehicle 1 and vehicle 2 rotate with the same angular velocity $w_2$ and have the same instantaneous phase, while vehicle 3 performs a combined  motion of both translation and rotation generated by the dynamics $\dot x_3 = v_1\text{cos}(\theta_1) + w_1(y_3 - y_1), \dot y_3 = \frac{\gamma_{13}(t)}{y_3 - y_1} + v_1 \text{sin}(\theta_1) + w_1(x_3 - x_1)$.
	\end{itemize}
	Evidently, by using the general motion solution in the abstraction   system \eqref{eq:motion_generation}, all kinds of feasible motions for the networked heterogeneous vehicles can be obtained.

	\section{Heterogeneous multi-vehicle coordination in a leader-follower structure} \label{sec:leader_follower_coordination}
	In this section we extend the above results to a leader-follower vehicle framework. The leader-follower structure involves a directed tree graph that describes the interaction relation within each individual vehicle, and has been used as a typical and benchmark framework in multi-vehicle coordination control (see e.g., \cite{panagou2013viability}). 
	
	\subsection{Formulation and feasibility of leader-follower constrained coordination}
	In a leader-follower structure, each follower vehicle has only one leader, whose motions constrain its immediate followers. For each (directed) edge $(i, j)$ we associate a vector  function $\Phi_{ij}(p_i, p_j, t)$ to describe equality coordination constraints. Different to the undirected graph modelling in the previous section,  here only the follower vehicle $j$ is responsible to maintain the equality constraint   $\Phi_{ij} = 0$ associated with edge $(i,j)$, and the leader vehicle $i$  is not affected by the equality constraint $\Phi_{ij}$. 
	
	The equality constraint $\Phi_{ij}(p_i, p_j, t)$ for edge $(i, j)$ is enforced along the trajectories of vehicles $i$ and $j$ if and only if $\Phi_{ij}(0) = 0$ and $\dot \Phi_{ij}(t) = 0, \forall t>0$. This gives 
	\begin{align}
		\frac{\text{d}\Phi_{ij}(t)}{\text{d} t} = \frac{\partial \Phi_{ij}(t)}{\partial p_i} \dot p_i +   \frac{\partial \Phi_{ij}(t)}{\partial p_j} \dot p_j  + \frac{\partial \Phi_{ij}(t)}{\partial t} = 0
	\end{align}
	Therefore, to enforce the equality constraint in the directed edge $(i, j)$, vehicle $j$'s motion should satisfy
	\begin{align}
		\frac{\partial \Phi_{ij}(t)}{\partial p_j} \dot p_j  = - \frac{\partial \Phi_{ij}(t)}{\partial p_i} \dot p_i  - \frac{\partial \Phi_{ij}(t)}{\partial t}    
	\end{align}

	Now we further consider the inequality constraint $\mathcal{I}_{ij}(p_i, p_j, t)$ associated with the edge $(i, j)$, while the follower vehicle $j$ is responsible to take care of the inequality constraint $\mathcal{I}_{ij}(p_i, p_j, t) \leq 0$. Suppose at time $t$ the inequality constraint is active in the sense that $\mathcal{I}_{ij}(p_i, p_j, t) = 0$. By the temporal viability theory and temporal set-invariance control in Lemma~\ref{theorem:invariant_temporal_set_affine}, vehicle $j$'s motion should satisfy 
	\begin{align}
		\frac{\text{d}\mathcal{I}_{ij}(t)}{\text{d} t} = \frac{\partial \mathcal{I}_{ij}(t)}{\partial p_i} \dot p_i +   \frac{\partial \mathcal{I}_{ij}(t)}{\partial p_j} \dot p_j + \frac{\text{d}\mathcal{I}_{ij}(t)}{\text{d}t} \leq 0
	\end{align}
	or equivalently 
	\begin{align}
		\frac{\partial I_{ij}(t)}{\partial p_j} \dot p_j \leq -   \frac{\partial I_{ij}(t)}{\partial p_i} \dot p_i   - \frac{\text{d}\mathcal{I}_{ij}(t)}{\text{d}t}
	\end{align}
	
	Further note that vehicle $j$'s motion is subject to the kinematic constraint described by an equivalent codistribution \eqref{eq:dynamics_constraint_transformation}
	\begin{align} 
		\Omega_{K_j} (p_j) \dot p_j = T_{K_j} (p_j).
	\end{align}
	
	To summarize, the condition for feasible coordination for a leader-follower vehicle team is stated as follows.
	\begin{theorem} \label{theorem:leader_follower}
		The coordination task for a leader-follower vehicle team with both  equality constraints and inequality constraints has feasible motions if and only if, for all follower vehicles $j = 2, \cdots, n$, the following \textit{mixed} differential-algebraic (in)equalities have solutions
		\begin{subequations}
			\begin{align}
				\Omega_{K_j}(p_j)  \dot p_j &= T_{K_j}(p_j)   \\
				\frac{\partial \Phi_{ij}(t)}{\partial p_j} \dot p_j  &= - \frac{\partial \Phi_{ij}(t)}{\partial p_i} \dot p_i  - \frac{\partial \Phi_{ij}(t)}{\partial t}, (i,j) \in \mathcal{E}    \\
				\frac{\partial I_{ij}(t)}{\partial p_j} \dot p_j &\leq -   \frac{\partial I_{ij}(t)}{\partial p_i} \dot p_i,\,\,\, \text{if}\, (i, j) \in \chi(P, t)
			\end{align} 
		\end{subequations}
		where $\chi(P, t)$ denotes the set of active constraints among all the edges at time $t$. 
	\end{theorem}

	To determine motion feasibility of the coordination task for the whole team, by following Algorithm \ref{algorithm:feasibility_undirected}, a recursive procedure can be performed embracing  all the vehicles in the leader-follower group, starting from the top leader to the last follower in the underlying directed tree graph. We note that in contrast to the undirected graph case,   such a recursive procedure for the directed tree graph for a leader-follower team enables a \textit{decentralized}  checking of the feasibility condition for each vehicle where the codistribution matrix $\Omega_K, \Omega_E, \Omega_I$ only involves vehicle~$j$ and the associated edge $(i, j)$ in Algorithm~\ref{algorithm:feasibility_undirected}, and the procedure can be terminated within a finite number of steps.

	\subsection{Simulation example: visibility maintenance coordination in leader-follower vehicle group}
	\subsubsection{Unicycle coordination with time-varying constraints} 
	In the context of temporal viability regulation, it may be useful to restrict certain parts of the joint state to desired state trajectories. As in Section~\ref{sec:complexleaderfollower}, let $v_{r}(t), u_{r}(t)$ define a desired reference trajectory for the first vehicle ($i=1$), expressed by the annihilating codistribution
	$$
	\cos(\theta_1)\text{d}x_1 + \sin(\theta_1)\text{d}y_1 = v_{r}(t), \quad \text{d}\theta_1 = u_{r}(t),
	$$
	here written compactly with $\Omega_{E}\dot{P} = T_{E}$, where then
	\begin{equation}\label{eq:timevaryingVelocity}
		\Omega_{E}\hspace{-1pt}=\hspace{-1pt}
		\begin{bmatrix}
			\cos(\theta_1)\text{d}x_1 \hspace{-1pt}+\hspace{-1pt} \sin(\theta_1)\text{d}y_1 \\ 
			\text{d}\theta_1
		\end{bmatrix}
		,\;
		T_{E}\hspace{-1pt}=\hspace{-1pt}
		\begin{bmatrix}
			v_{r}(t)\\
			u_{r}(t)
		\end{bmatrix}.
	\end{equation}
	This can be replaced by a general vector field for more complex applications. For instance, we can  compute a cell decomposition as described in \protect{\cite[Appendix C]{panagou2014cooperative}} along with a cell potential as described in~\protect{\cite[Theorem 1]{lindemann2009simple}}, thereby enabling movement though cluttered environments. However, to illustrate the time-varying constraint satisfaction, we here keep the example minimal with specific $(v_r,u_r)$.
	
	
	\subsubsection{Time-varying distance inequality constraints}
	As previously demonstrated, it is often useful to pose inequality constraints on inter-vehicular distances in mobile vehicle coordination tasks. By Theorem~\ref{theorem:leader_follower}, we can consider time-varying constraints in the form
	\begin{align}
		\frac{1}{2}d_{ij}^{\mathrm{min}}(t)^2\leq \frac{1}{2}(x_i - x_j)^2 +  \frac{1}{2}(y_i - y_j)^2  \leq \frac{1}{2}d_{ij}^{\mathrm{max}}(t)^2
	\end{align}
	for some time-varying lower and upper bounds $0\leq d_{ij}^{\mathrm{min}}(t)<d_{ij}^{\mathrm{max}}(t)$, which may be equivalently  expressed as a vector-valued inequality
	$$
	g_{ij}^d(P, t) := \frac{1}{2}
	\begin{bmatrix}
		+(x_i - x_j)^2 +  (y_i - y_j)^2 - d_{ij}^{\mathrm{max}}(t)^2\\
		-(x_j - x_i)^2 -  (y_i - y_j)^2 + d_{ij}^{\mathrm{min}}(t)^2\\
	\end{bmatrix} \leq 0.
	$$
	The utility of such a constraint cannot be understated, as the time-varying constraints may be dynamically updated when interacting with the environment, the only requirement being that $d^{\mathrm{min}}_{ij}(t)$ and $d^{\mathrm{max}}_{ij}$ are  continuously differentiable in time. The constraint can be written on the form of Theorem~\ref{theorem:leader_follower}, as $\Omega_I^d \dot{P} \leq T_I^d$ where
	\begin{align}
		\Omega_{I,ij}^d &:= \nabla_P^\top g_{ij}^d(P, t) \nonumber \\
		&=
		\begin{bmatrix}
			(x_i-x_j)(\text{d}x_i-\text{d}x_j) + (y_i-y_j)(\text{d}y_i-\text{d}y_j)
			\\
			(x_j-x_i)(\text{d}x_i-\text{d}x_j)  + (y_j-y_i)(\text{d}y_i-\text{d}y_j)
		\end{bmatrix},  \nonumber \\ \label{eq:timevaryingDistance}
		T_{I,ij}^d &:= -\nabla_t  g_{ij}^d(P, t)= 
		\begin{bmatrix}
			+d_{ij}^{\mathrm{max}}(t)\dot{d}_{ij}^{\mathrm{max}}(t)
			\\
			-d_{ij}^{\mathrm{min}}(t)\dot{d}_{ij}^{\mathrm{min}}(t)
		\end{bmatrix},
	\end{align}
	using the standard dual bases $\{\text{d}x_i, \text{d}y_i, \text{d}\theta_i, \text{d}x_j, \text{d}y_j, \text{d}\theta_j\}$ in representing the annihilating codistribution associated with the gradient. 
	\subsubsection{Time-varying visibility inequality constraints}
	Similarly, we can consider time-varying visibility constraints posed on the relative rotation of vehicles, presented in a time-invariant form in~\cite{sun2019feasible} which were used in Section~\ref{sec:complexleaderfollower}. Here, the cosine angle of the body direction of system $j$, $b_{j} := [\cos(\theta_j), \sin(\theta_j)]$, and the  direction of system $i$ relative to $j$, $a_{ij} := [x_i - x_j, y_i - y_j]$ is bounded (see Fig.~\ref{fig:geometry}). In other words, we enforce a constraint on the time-varying apex angle, $2\alpha_{ij}(t)$, defining a cone of visibility of system $j$, such that
	\begin{align}\label{eq:cosineideas}
		g_{ij}^v(P, t) := \cos(\alpha_{ij}(t)) - \frac{\langle a_{ij},b_{j} \rangle}{\langle a_{ij}, a_{ij} \rangle^{1/2}}\leq  0.
	\end{align}
	By letting $c_{j} := [-\sin(\theta_j), \cos(\theta_j)]$, the associated annihilating codistribution of this time-varying inequality constraint becomes
	\begin{align}
		\Omega_{I,ij}^v\hspace{-1pt}:=&\nabla_P^\top g_{ij}^v(P, t) \nonumber \\
		=&\hspace{-1pt}\dfrac{\langle a_{ij}, c_j\rangle}{\sqrt{\langle a_{ij}, a_{ij} \rangle}} \Bigg(
		\dfrac{1}{\langle a_{ij}, a_{ij} \rangle} \Bigg\langle a_{ij},
		\hspace{-3pt}
		\Bigg[
		\begin{array}{c}
			\hspace{-5pt}\text{d}x_i-\text{d}x_j\hspace{-5pt}\\
			\hspace{-5pt}\text{d}y_j-\text{d}y_i\hspace{-5pt}
		\end{array}
		\Bigg]
		\hspace{-1pt}
		\Bigg\rangle
		\hspace{-2pt}+
		\hspace{-2pt}\text{d}\theta_j\Bigg).\nonumber
		\\\label{eq:timevaryingAngle}
		T_{I,ij}^v :=& -\nabla_t  g_{ij}^v(P, t) = \sin(\alpha_{ij}(t))\dot{\alpha}_{ij}(t).
	\end{align}
	With these general descriptions of the annihilating codistributions associated with the time-varying constraints, we proceed to show how they may be implemented in a numerical example.
	
	\begin{figure}[t]
		\begin{center}
			\includegraphics[width=0.40\textwidth]{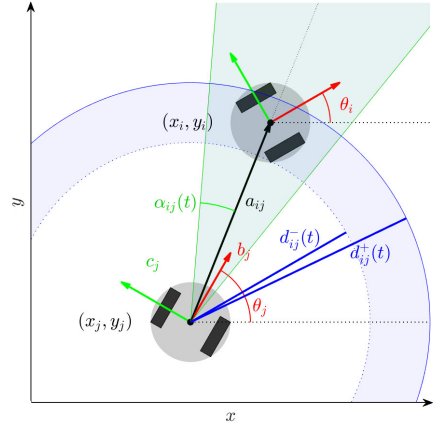}
			\vspace{-15pt}
			\caption{Illustration of a visibility inequality constraint, ${g}^{v}_{ij}(p,t)$, bounding the direction $b_j$ to the green cone defined by the vector $a_{ij}$ and the half apex angle $\alpha_{ij}$, and the two-sided distance constraint (blue).}
			\label{fig:geometry}
		\end{center}
	\end{figure}
	
	\subsubsection{Simulation demonstration} 
	
	We now tackle the challenging task of controlling a multi-vehicle coordination system consisting of $n=3$ vehicles, where vehicle $i=1$ is a car vehicle obeying the kinematics in~\eqref{eq:car_model}, and the remaining two vehicles are unicycles, as described in~\eqref{eq:example_unicycle}. A set of time-varying inequality constraints are posed on the joint state of the vehicles, $P(t)\in \mathbb{R}^{10}$. This example is similar to that in Section~\ref{sec:complexleaderfollower}, but now lifting the constraints that vehicle $i=1$ should be visible from vehicle $i = 3$, and making all of the inequality constraints time-varying and using different kinematics. This is achieved by defining
	\begin{align*}
		\Omega_{K,i}\dot{P}&=0, &&i\in\{1,2,3\},&&(\text{see Section~\ref{sec:kinematic_model}})\\
		\Omega_{E,1}\dot{P}&=T_{E,1},&&&&
		(\text{see}~\eqref{eq:timevaryingVelocity})\\
		\Omega_{I,ij}^d\dot{P}&\leq T_{I,ij}^d,&&(i,j)\in\{(1,2),(2,3)\},&&(\text{see}~\eqref{eq:timevaryingDistance})\\
		\Omega_{I,ij}^v\dot{p}&\leq T_{I,ij}^v,&&(i,j)\in\{(1,2),(2,3)\},&&(\text{see}~\eqref{eq:timevaryingAngle})
	\end{align*}
	with the right hand side defined by eight arbitrary time-varying functions. For illustrative purposes, these functions are chosen as
	\begin{subequations}
		\begin{align*}
			(u_{1,1}(t),u_{1,2}(t))&=(1+0.1\sin(t),\sin(t)),\\
			(d^{\mathrm{min}}_{12}(t),d^{\mathrm{min}}_{23}(t))&=(1.0,1.0)+0.2\sin(t)^2 (-1,1),\\
			(d^{\mathrm{max}}_{12}(t),d^{\mathrm{max}}_{23}(t))&=(1.3,1.3)+0.2\sin(t)^2 (-1,1),\\
			(\alpha_{12}(t),\alpha_{23}(t))&=(0.05, 0.1) + 0.1(\cos(t)^2, \sin(t/2)^2).
		\end{align*}
	\end{subequations}
	
	In effect, this set of constraints result in the vehicles satisfying the associated kinematic equations of a car in Equation~\eqref{eq:car_model} and the unicycle in Equation~\eqref{eq:example_unicycle}, while  vehicle $i=1$ follows a path given by $u_{1,1}(t), u_{1,2}(t)$; vehicle $i=2$ maintains a set of time-varying inequality constraints with respect to the vehicle $i=1$; and vehicle $i=3$ maintains a set of time-varying inequality constraints with respect to the vehicle $i=2$. In this demonstration, we exemplify  a heterogeneous leader-follower scenario defined using inequality constraints (here permitted to be time-varying) instead of the more common equality constraints.

	\begin{figure}[t]
		\begin{center}
			\includegraphics[width=0.8\columnwidth]{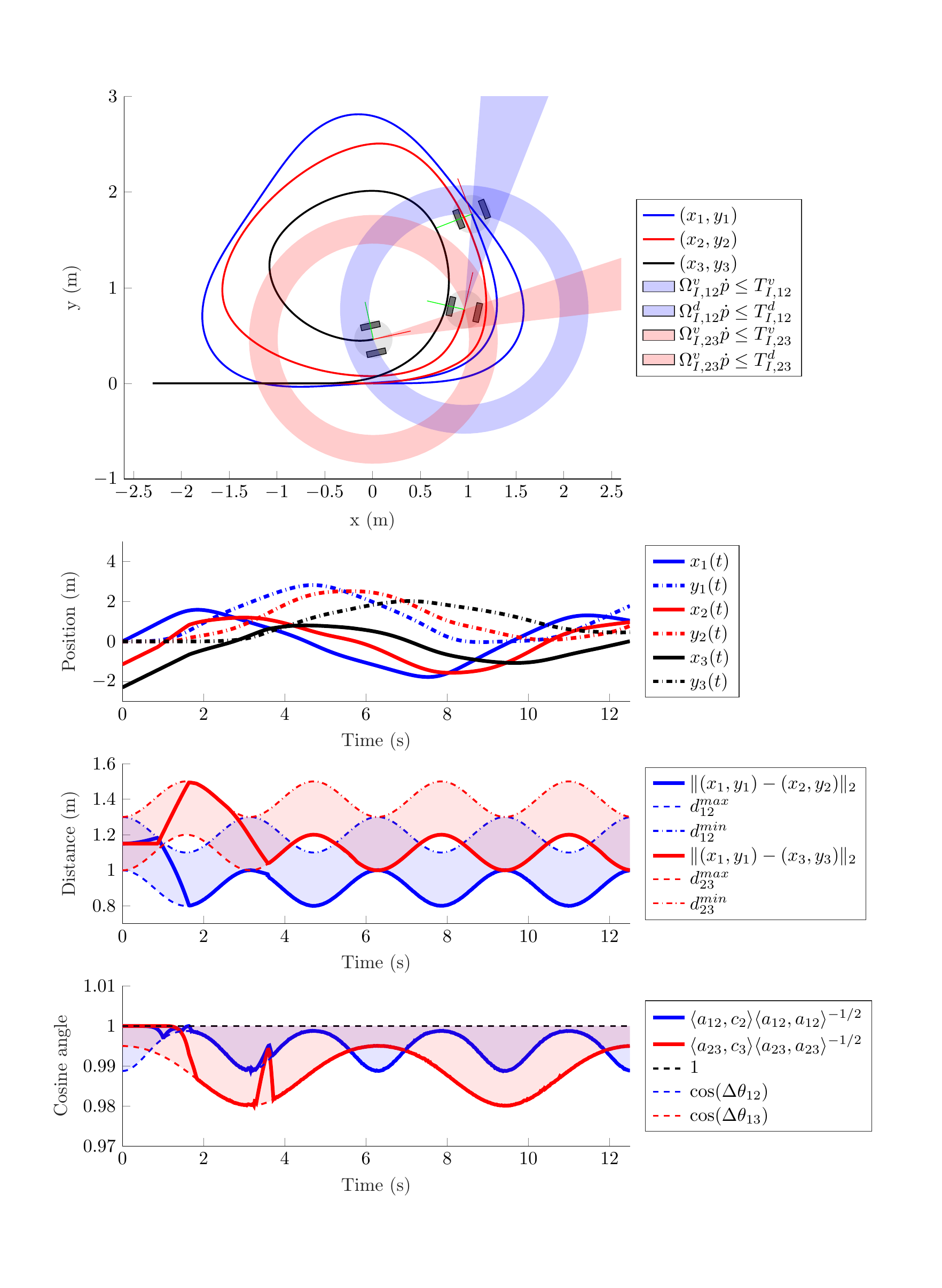}
			\caption{Complex vehicle coordination in a leader-follower structure. \textit{Top:} Positional trajectories of the vehicles, $i=1$ (red), $i=2$ (blue) and $i=3$ (black), with the inequality constraints in the relative distances and visibility depicted at the terminal time $t=12.5$. \textit{Top, center:} The positional trajectories of the vehicles in time. \textit{Bottom, center:} The time-varying distance inequality constraints associated with the vehicles $(i,j)=(1,2)$ (blue) and $(i,j)=(2,3)$ (red). \textit{Bottom:} The narrow time-varying visibility inequality constraints (c.f. Figure~\ref{fig:complexleaderfollower}) associated with vehicles $(i,j)=(1,2)$ (blue) and $(i,j)=(2,3)$ (red). A video of the simulation can be accessed through~\cite{greiff2021video}.}
			\label{fig:tvsimulation}
		\end{center}
	\end{figure}
	
	For every possible combination of active inequality constraints, the resulting set of algebraic equations in Theorem~\ref{theorem:invariant_temporal_set_affine} is solved symbolically offline. In our case, there is a total of six inequality constraints, generating many combinations of possible active constraints. Each of the corresponding solutions will be associated with varying degrees of freedom, $l\in [0,4]$, in which we may choose any $W(t)=[w_1(t),...,w_l(t)]^{\top}\in \mathbb{R}^l$ such that Theorem~\ref{theorem:invariant_temporal_set_affine} is satisfied at a time $t$. In the simulation, we restrict $W(t)$ to a compact set and simply pick an admissible control   which minimizes $\|(\mathrm{d}^2/\mathrm{d}t^2)P(t)\|_{2}$ subject to the active algebraic constraints, similar to the example in Section~\ref{sec:complexleaderfollower}. Studying the simulation result in Figure~\ref{fig:tvsimulation}, it is clear that the time-varying inequality constraints on the inter-vehicle distance and the visibility are met at all times. The particular choice of   $\alpha_{12}(t)$ generates an extremely narrow cone of visibility at certain points in time, which makes it more challenging to generate constrained coordinated motions. Despite this, all of the constraints are satisfied at all times, and we also note that one or more of the inequality constraints are saturated at almost all times. 
	\section{Conclusions} \label{sec:conclusions}
	In this paper, we discuss the cooperative constrained motion coordination  problem for multiple mobile heterogeneous vehicles subject to many different constraints (holonomic/nonholonomic motion constraints, geometric coordination  constraints, equality or inequality constraints, among others). Using tools from differential geometry, distribution/codistributions for control-affine systems and viability theory, we have developed a general framework to determine whether feasible motions exist for a multi-vehicle group that meet both heterogeneous kinematic constraints and various coordination constraints. The coordination task is encoded by a mix of inequality and equality functions for describing a coordination task, which  typically include distance-constrained formation, vehicle heading coordination,  and visibility maintenance, etc.
	A motion generation algorithm is proposed to find feasible motions and trajectories for heterogeneous vehicles to achieve a coordination task. 
	
	To address time-varying coordination tasks, we have developed a general theory on temporal viability and controlled invariant set with temporal functions. We then extend the framework to deal with the cooperative coordination problem with time-varying coordination tasks and leader-follower vehicle structure.  
	Several case study examples and simulation experiments involving the coordination of multiple unicycle vehicles, constant-speed nonholonomic vehicles, fully-actuated vehicles, and car-like vehicles are provided to illustrate the proposed coordination control schemes.   In the future, we will investigate cooperative motion \textit{stabilization} control for heterogeneous vehicles under various physical and geometrical constraints (e.g., saturated control input and combined input-state constraints). We are currently looking into the extension of the proposed coordination framework to address  collaborative motion control of multiple fixed-wing systems \cite{sun2021collaborative} under speed constraints. Other interesting future directions include more efficient
	algorithmic techniques for cooperative motion generation, and distributed and optimal motion coordination for networked heterogeneous systems.    
	

	\bibliography{Coordination}
	\bibliographystyle{ieeetr}

	\begin{small}

		\begin{IEEEbiography}[{\includegraphics[width=1in,height=1.25in,clip,keepaspectratio]{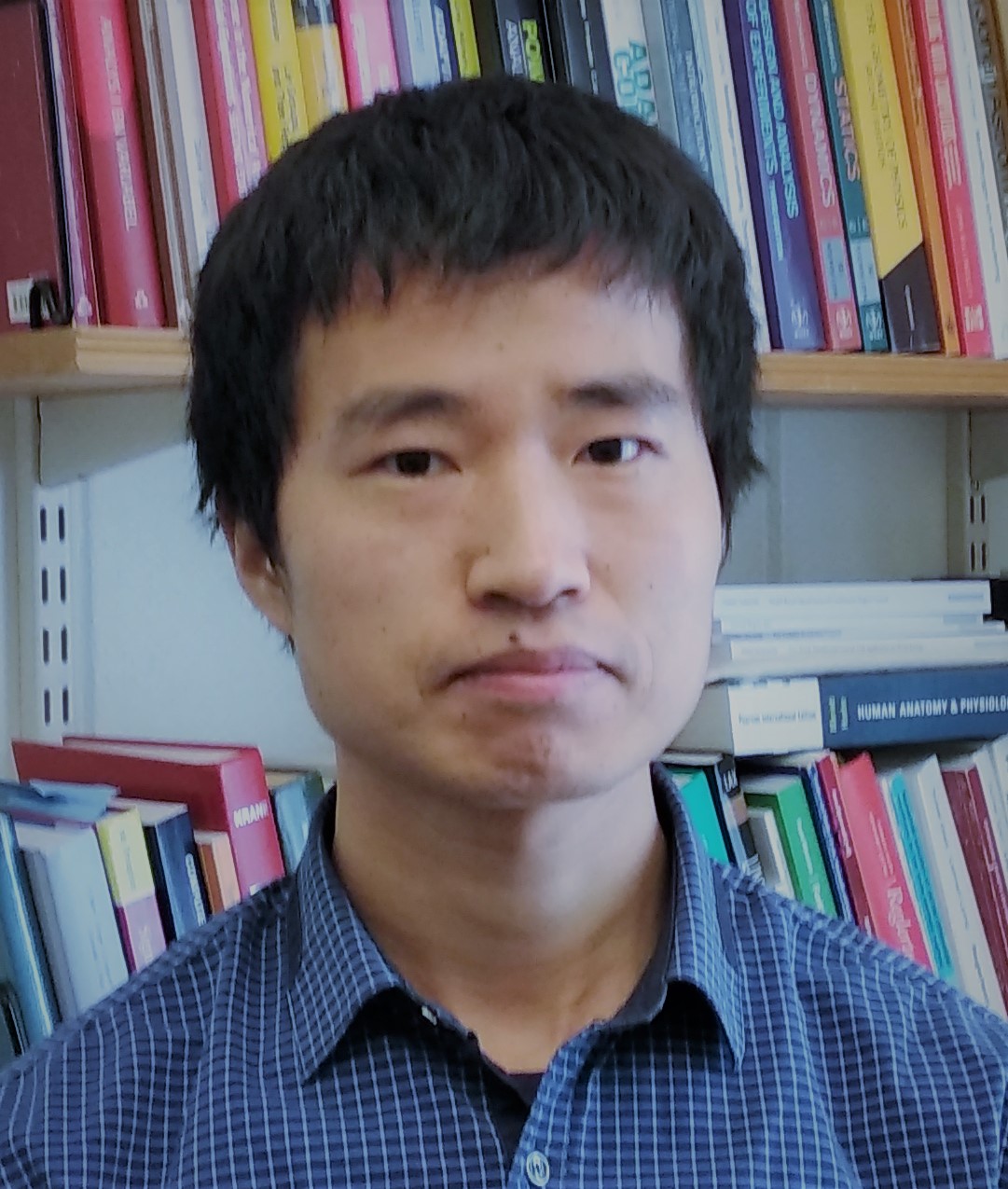}}]{Zhiyong Sun} (Member, IEEE) received the Ph.D. degree from The Australian National University (ANU), Canberra ACT, Australia, in February 2017. He was a Research Fellow/Lecturer with the Research School of Engineering, ANU, from 2017 to 2018. From June 2018 to January 2020, he worked as a postdoctoral researcher at Department of Automatic Control, Lund University, Lund, Sweden. Since January 2020 he has joined Eindhoven University of Technology (TU/e), the Netherlands, as an assistant professor. His research interests include multi-robotic systems, control of autonomous formations, distributed control and optimization.
			
			Dr. Sun was a recipient of the Australian Prime Minister’s Endeavour Postgraduate Award in 2013 from the Australian Government, and the Outstanding
			Overseas Student Award from the Chinese Government in 2016.
			He  received a Best Student Paper Finalist Award from the 54th IEEE Conference on Decision and Control (CDC) at Osaka, Japan, and was the winner of the Best Student Paper Award from the 5th Australian Control Conference at Gold Coast, Australia. He was awarded the Springer PhD Thesis Prize from Springer in 2017, and has authored the book titled \textit{Cooperative Coordination and Formation Control for Multi-agent Systems} (Springer, 2018).
		\end{IEEEbiography}
		
		\begin{IEEEbiography}[{\includegraphics[width=1in,height=1.25in,clip,keepaspectratio]{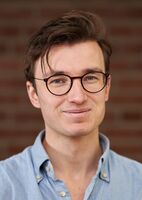}}]{Marcus Greiff
			} (Student member, IEEE) received the Ph.D. degree from the Department of Automatic Control at Lund University   in November 2021. He is  affiliated with Wallenberg AI, Autonomous Systems and Software Program (WASP), and currently working as a part-time consultant for MERL in Boston.
			
			His research topics concern nonlinear control and estimation in the context of quad-rotor UAV control. He is the winner of the best student paper award at the 2020 IEEE Conference on Control Technology and Applications (CCTA). 
		\end{IEEEbiography}
		
		\begin{IEEEbiography}[{\includegraphics[width=1in,height=1.25in,clip,keepaspectratio]{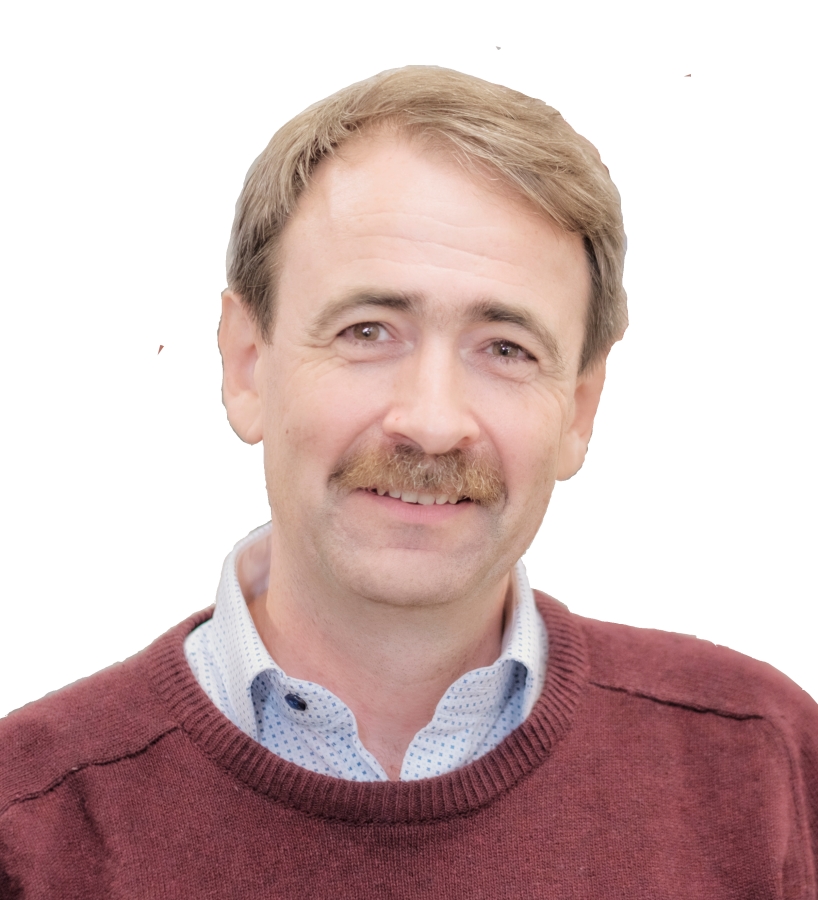}}]{Anders Robertsson} (Senior Member, IEEE) received the M.Sc. degree in electrical engineering and the Ph.D. degree in automatic control from LTH, Lund University, Lund, Sweden, in 1992 and 1999, respectively. He was appointed as a Docent in 2005 and an Excellent Teaching Practitioner in 2007. Since 2012, he has been a Full Professor with the Department of Automatic Control, LTH, Lund University. His research interests include nonlinear estimation and control, robotics research on mobile and industrial manipulators, and feedback control of computing systems.
		\end{IEEEbiography}
		
		\begin{IEEEbiography}[{\includegraphics[width=1in,height=1.25in,clip,keepaspectratio]{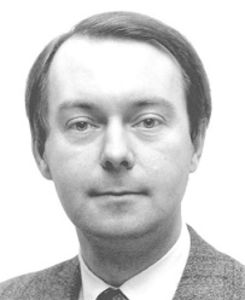}}]{Rolf Johansson} (Fellow, IEEE) received the
			M.Sc. degree in technical physics, the bachelor of-medicine degree, the Doctorate degree in
			control theory, and the M.D. degree from Lund
			University, Lund, Sweden, in 1977, 1980, 1983,
			and 1986, respectively. Since 1986, he has been with the Department
			of Automatic Control, Lund University, where he
			is currently a Professor of control science. In his
			scientific work, he has been involved in research
			in adaptive system theory, mathematical modeling, system identification, biomathematics, robotics, and combustion
			engine control.
			
			Dr. Johansson is a Fellow of the Royal Physiographic Society, Section
			of Medicine. He received the 1995 Ebeling Prize of the Swedish Society
			of Medicine for distinguished contribution to the study of human balance.
			Since 1999, he has been an Associate Editor of the \textit{International Journal
				of Adaptive Control and Signal Processing}. He is also an Editor of
			\textit{Intelligent Service Robotics} and a Member of the editorial boards of
			\textit{Mathematical Biosciences}.
		\end{IEEEbiography}
		
		\begin{IEEEbiography}[{\includegraphics[width=1in,height=1.25in,clip,keepaspectratio]{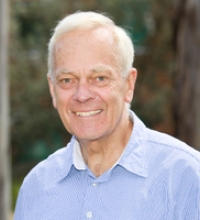}}]{Brian D. O. Anderson}  (M’66-SM’74-F’75-LF’07) was born in Sydney, Australia, and educated at Sydney University in mathematics and electrical engineering, with PhD in electrical engineering from Stanford University in 1966. 
			
			He is an Emeritus Professor at the Australian National University (having retired as Distinguished Professor in 2016), and was formerly Distinguished Professor at Hangzhou Dianzi University, and Distinguished Researcher in Data-61 CSIRO, Australia. His awards include the IEEE Control Systems Award of 1997, the 2001 IEEE James H Mulligan, Jr Education Medal, and the Bode Prize of the IEEE Control System Society in 1992, as well as several IEEE and other best paper prizes. He is a Fellow of the Australian Academy of Science, the Australian Academy of Technological Sciences and Engineering, the Royal Society, and a foreign member of the US National Academy of Engineering. He holds honorary doctorates from a number of universities, including Universite Catholique de Louvain, Belgium, and ETH, Zurich. He is a past president of the International Federation of Automatic Control and the Australian Academy of Science. His current research interests are in distributed control, social networks and econometric modelling.
		\end{IEEEbiography}
	\end{small}

\end{document}